\documentclass[11pt,reqno]{amsart}  

\usepackage{amscd} 
\usepackage{mathrsfs}
\usepackage{amsmath,amsbsy,amsthm,amsfonts,amssymb,esint}
\usepackage[unicode,psdextra]{hyperref}
\usepackage{graphicx}
\usepackage{tikz}
\usepackage{bm}
\usepackage{soul}
\usepackage{color}

\usepackage[sort&compress,capitalise,nameinlink]{cleveref}
\crefname{section}{\textsection}{\textsection}
\crefname{subsection}{\textsection}{\textsection}
\usepackage{dsfont}
\usepackage{enumerate}

\DeclareMathAlphabet\mathbfcal{OMS}{cmsy}{b}{n}

\newtheorem{theorem}{Theorem}[section]
\newtheorem{lemma}[theorem]{Lemma}
\newtheorem{proposition}[theorem]{Proposition}
 
\theoremstyle{definition}

\newtheorem{remark}[theorem]{Remark} 
\numberwithin{equation}{section}

\DeclareMathOperator*{\esssup}{ess\,sup}

\newcommand{\one}{\mathds{1}}

\newcommand{\xv}{\mathbf{x}}

\newcommand{\dom}{\mathscr{D}}
\newcommand{\diff}{\mathrm{d}}
\newcommand{\eps}{\varepsilon}

\newcommand{\supp}{\mathrm{supp}\,}
\newcommand{\avg}[1]{\lf\langle #1 \ri\rangle}

\newcommand{\beq}{\begin{equation}}
\newcommand{\eeq}{\end{equation}}

\newcommand{\dist}{\mathrm{dist}}

\newcommand{\R}{\mathbb{R}}
\newcommand{\N}{\mathbb{N}}

\newcommand{\bdm}{\begin{displaymath}}
\newcommand{\edm}{\end{displaymath}}
\newcommand{\bdn}{\begin{eqnarray}}
\newcommand{\edn}{\end{eqnarray}}
\newcommand{\bay}{\begin{array}{c}}
\newcommand{\eay}{\end{array}}
\newcommand{\ben}{\begin{enumerate}}
\newcommand{\een}{\end{enumerate}}
\newcommand{\beqn}{\begin{eqnarray}}
\newcommand{\eeqn}{\end{eqnarray}}
\newcommand{\bml}[1]{\begin{multline} #1 \end{multline}}
\newcommand{\bmln}[1]{\begin{multline*} #1 \end{multline*}}

\newcommand{\lf}{\left}
\newcommand{\ri}{\right}
\newcommand{\disp}{\displaystyle}
\newcommand{\tx}{\textstyle}

\newcommand{\braketr}[2]{\lf\langle #1\lf|#2\ri. \ri\rangle}
\newcommand{\braketl}[2]{\lf.\lf\langle #1\ri|#2 \ri\rangle}

\newcommand{\meanlrlr}[3]{\lf\langle #1\lf|#2\ri|#3\ri\rangle}

\renewcommand{\leq}{\leqslant}
\renewcommand{\geq}{\geqslant}

\DeclareMathOperator*{\slim}{s\,-\,lim}

\title[Schr\"odinger operators with multiple Aharonov-Bohm fluxes]{Schr\"odinger operators with multiple Aharonov-Bohm fluxes}

\author{Michele Correggi}
\address{Dipartimento di Matematica, Politecnico di Milano, P.zza Leonardo da Vinci, 32, 20133, Milano, Italy}
\email{michele.correggi@gmail.com}
\urladdr{https://sites.google.com/view/michele-correggi}

\author{Davide Fermi}
\address{Dipartimento di Matematica, Politecnico di Milano, P.zza Leonardo da Vinci, 32, 20133, Milano, Italy\\
and Istituto Nazionale di Fisica Nucleare, Sezione di Milano, Italy}
\email{davide.fermi@polimi.it}
\urladdr{https://fermidavide.com}

\begin{document}

\begin{abstract} 
	 We study the Schr\"{o}dinger operator describing a two-dimensional quantum particle moving in presence of $ N \geqslant 1$  Aharonov-Bohm magnetic fluxes. We classify all the self-adjont realizations of such an operator, providing an explicit characterization of their domains and actions. Moreover, we examine their spectral and scattering properties, proving in particular the existence and completeness of wave operators in relation with the free dynamics.
\end{abstract}


\maketitle

\section{Introduction and Main Results}

Back in 1959, Aharonov and Bohm \cite{AB59} predicted that the presence of a magnetic field would induce a phase shift in the wave function of a charged quantum particle, even if the particle is confined for all times to a space region where the magnetic field vanishes identically. Analogous theoretical considerations had previously been advanced by Ehrenberg and Siday \cite{ES49}. 
Undisputed experimental evidence of the Aharonov-Bohm (AB) effect was provided in 1986 by Tonomura {\it et al.} \cite{TO86}, who used micro-sized toroidal magnets coated in superconducting layers to minimize leakages of the magnetic field and thus realize perfect shielding for the electron wave function. Though a controversy about the fundamental interpretation of the AB effect somehow persists even nowadays \cite{K22,MV20}, the reality of the physical phenomenon is by now unquestionable. In fact, it appears to play a prominent role in many areas of condensed matter physics.
For example, sharply localized fluxes of AB-type occur in {\it type-II superconductors}, when strong magnetic fields pierce through an almost 2D layer \cite{Ab57}. Another prominent instance regards {\it anyons}, quasi-particles excitations of 2D electron gases in the fractional Quantum Hall regime which carry fractional charge and obey fractional statistics \cite{BK20,LM77,NL20,Wi82}. Notably, with a suitable gauge fixing, anyons can be described in terms of bosonic wave functions with 2-body AB interactions. Yet another research line, which is recently attracting attention, is that of {\it AB cages}, where the destructive interference produced by singular magnetic fluxes provides a localization mechanism for non-interacting particles \cite{BPC21,MVV22}.

With the above applications in mind, we consider a prototype model including a charged scalar particle and $N$ parallel ideal solenoids, each of infinite length and zero diameter like in the original AB setting. After factorization of the axial direction, the dynamics of the particle is determined by a Schr\"odinger operator in $\mathbb{R}^2$ of the form
	\begin{equation}\label{eq: HNintro}
		\left(- \,i \nabla + \sum_{n \,=\, 1}^{N}\alpha_n\; \frac{\lf( \xv - \xv_n \ri)^{\perp}}{\mathbf{x} - \mathbf{x}_{n}|^2} \right)^{\!2} .
	\end{equation}
Here, $ \xv = (x,y) \in \mathbb{R}^2 $ and $\mathbf{x}^{\perp} = (- y, x) $, while $ \xv_n $ and $\alpha_n \in \mathbb{R}$ identify the positions of the solenoids and the associated magnetic fluxes, respectively. It is worth noting that the magnetic field matching the vector potential in \eqref{eq: HNintro} is formally given by a sum of Dirac delta distributions, namely,
	\begin{equation*}
		B(\mathbf{x}) =  \sum_{n \,=\, 1}^{N}\, 2\pi \,\alpha_n\,\delta(\mathbf{x} - \mathbf{x}_{n})\,.
	\end{equation*}
Hence, $B$ has compact support and we can reasonably expect that scattering w.r.t. the free Laplacian is well-defined. 
In light of the above physical interpretation, we may also think of the formal operator introduced above as describing in suitable regimes a tracer particle moving in a gas of $ N $ anyons with heavy masses (or two-dimensional quasi-particles carrying localized magnetic fluxes), so that the latters may be modeled by fixed AB fluxes at the particles' positions $ \xv_n $, $ n = 1, \dots, N $ \cite{CDLR19,CLR17,CO18,Gi20,LR15,LR16,LS14,Ma96}.

The configuration with just one solenoid has an evident rotation symmetry and, as a consequence, it can be explicitly analyzed by decomposition in angular harmonics. Building on this, an exhaustive classification of all self-adjoint realizations in $L^2(\mathbb{R}^2)$ of the Schr\"odinger operator \eqref{eq: HNintro} with $N = 1$ was first derived using Krein-von Neumann theory by Adami and Teta \cite{AT98} (see also \cite{DS98} for a similar result, \cite{BDG11,DF21,DFNR20,DR17,DG20,DG21} for more details on the radial operators, and \cite{CO18,Ou94,MMO04} for connections with 2-anyons systems). These single-flux Hamiltonians include the Friedrichs realization and a 4-parameter family of singular perturbations thereof, describing \textsl{s-wave} and \textsl{p-wave} zero-range interactions. 
The spectral and scattering properties of these operators were further investigated in \cite{PR11,Ru83,Ta99,Y21}. General results on the scattering matrix and total cross section for Schr\"odinger operators related to the single flux Friedrichs Hamiltonian were derived in \cite{RY02,Ya03,Ya06} by means of stationary representation formulas and pseudo-differential methods.
	
On the other hand, the approximation in resolvent sense of the above mentioned Hamiltonians with zero-range interactions by means of non-singular electromagnetic potentials was discussed in \cite{MVG95,OP08,Ta01}. These works shed some light on the idealizations understood in the original AB set up, indicating in particular that Schr\"{o}dinger operators comprising zero-range potentials naturally emerge when considering resonant shielding potentials. On the other hand, the effects produced by regular magnetic perturbations have also been examined. The case of a uniform magnetic field, on top of the AB singularity, was studied in \cite{ESV02}, again by Krein-von Neumann methods. A different approach based on quadratic form techniques was first proposed in \cite{CO18}, and later extended in \cite{CF21,F22} to encompass generic, regular magnetic perturbations.
Let us also mention that, in the single flux setting, there are classical results on Hardy-type inequalities \cite{LW99,BDELL20} and dispersive estimates \cite{GK14}. Another research line regards studying the AB Hamiltonian in compact domains, examining the behavior of simple eigenvalues under variations of the flux position \cite{AFNN18,AN18} (see also \cite{FNOS23} for similar results in configurations with many coalescing poles).

Investigating cases with more than one flux ($N \geqslant 2$) is in general a harder task, due to the lack of exact solutions. Remarkably, a series expansion for the Green function associated to a 1D array of fluxes was derived in \cite{St89,St91}, by universal covering space techniques. Similar methods were employed in \cite{Na00,FG08} to produce some explicit formulas for the solutions of the eigenvalue problem. These results apparently support a previous conjecture of Aharonov, suggesting that an array of flux tubes would act as a repulsive barrier for low-energy particles.
Again by complex analysis techniques and direct computations, an analytic expression for the scattering amplitude for two opposite AB fluxes was obtained in \cite{BMO10}, focusing separately on the limit situations of small fluxes and small distance between the fluxes.

	The (semiclassical) scattering amplitude and the presence of resonances in the regime of large separation between two arbitrary singular fluxes were instead examined in \cite{AT11,AT14,IT01,Ta07,Ta08} by means of stationary methods.
	Regarding Schr\"odinger operators of the form \eqref{eq: HNintro} with more than two fluxes, let us mention that the leading order singularity of the wave trace and resonances for the wave propagator were studied in \cite{Y22}. We also highlight that a rigorous diamagnetic inequality for Schr\"odinger operators of the form \eqref{eq: HNintro} was derived in \cite{MOR04}, and used contextually to deduce Lieb-Thirring and CLR-type eigenvalue estimates.
Another valuable theoretical approach to the analysis of configurations with several magnetic fluxes is the mean field approximation. Making reference to this regime, by perturbative and numerical computations, the scattering from a periodic array of fluxes was examined in \cite{KW94} and the density of states in the thermodynamic limit was studied in \cite{DFO95,DFO97}.
	Finally, let us remark that the AB effect was also investigated for Schr\"odinger and Klein-Gordon dynamics with electromagnetic potentials in the exterior of compact obstacles. In particular, high-velocity estimates for the scattering operator and inverse scattering results were deduced by means of time dependent techniques in \cite{BW09,BW16}.

All the works on multiple fluxes mentioned above refer to the Friedrichs realization of the operator \eqref{eq: HNintro}. This means that the attention is always restricted to wave functions vanishing at the points $\mathbf{x}_{n}$, which amounts to consider only perfectly shielded solenoids. Nonetheless, the single flux studies cited before give a strong indication of the fact that other self-adjoint realizations of \eqref{eq: HNintro} in $L^2(\mathbb{R}^2)$ should exist and that they cannot be neglected in some specific physical contexts.
It goes without saying that a more faithful physical representation can be obtained taking also into account relativistic effects. In this regard, to describe the motion of a spin $1/2$ charged particle confined to a 2D slab, punctured by ideal solenoids, some authors have considered Pauli and Dirac operators analogous to the Schr\"odinger one introduced in \eqref{eq: HNintro}. The self-adjointness and spectral features of these operators were investigated in \cite{EV02,GS04,KA14,Pe05,Pe06}, focusing especially on the degeneracy of the zero-energy modes and on the related Aharonov-Casher formula.

In this work, we first characterize all the admissible self-adjoint realizations in $L^2(\mathbb{R}^2)$ of the Schr\"odinger operator \eqref{eq: HNintro} using quadratic form techniques and the Kre{\u\i}n theory of self-adjoint extensions. Next, we use classical resolvent arguments to investigate the spectral and scattering properties of the Hamiltonians thus obtained.

\bigskip

\noindent
{\footnotesize
\textbf{Acknowledgments}
This work has been supported by MUR grant Dipartimento di Eccellenza 2023--2027. MC acknowledges the support of PNRR Italia Domani and Next Generation Eu through the ICSC National Research Centre for High Performance Computing, Big Data and Quantum Computing. DF acknowledges the support of the European Research Council (ERC) under the European Union's Horizon 2020 research and innovation programme (ERC CoG {\it UniCoSM}, grant agreement n. 724939) and of INdAM-GNFM Progetto Giovani 2020 ``\textsl{Emergent Features in Quantum Bosonic Theories and Semiclassical Analysis}''.}

\subsection{The model: Schr\"{o}dinger operators with multiple Aharonov-Bohm fluxes}

As anticipated our main goal is to study the well-posedness as a self-adjoint operator of the Hamiltonian describing a quantum two-dimensional spinless particle moving in presence of $ N \geq 1 $ AB fluxes, {\it i.e.}, the formal Schr\"{o}dinger operator
	\begin{equation}\label{eq: HalN}
		H_{N} = \left(- \,i \nabla + \tx\sum_{n \,=\, 1}^{N}\mathbf{A}_{n} \right)^{\!2} , 
		\qquad 
		\mathbf{A}_{n}(\mathbf{x}) = \alpha_n\;\frac{(\mathbf{x} - \mathbf{x}_{n})^{\perp}}{|\mathbf{x} - \mathbf{x}_{n}|^2}\,,
	\end{equation}
where, for each $n \in \{1,\dots,N\}$, $\mathbf{x}_{n} \!\in\! \mathbb{R}^2$ identifies the position of the $n$\textsuperscript{th} singular flux and $\alpha_n \in (0,1)$ measures its intensity (see also the next \cref{rem: an1}).
A natural dense domain where the operator above makes sense is the set $ C^{\infty}_{\mathrm{c}}(\mathbb{R}^2\!\setminus\! \lf\{ \xv_1, \dots, \xv_N \ri\}) $, {\it i.e.}, smooth functions with support away from the fluxes. Notice that the vector fields $\mathbf{A}_{n}$ separately fulfill the Coulomb gauge, in the sense of Schwartz distributions:
	\begin{equation*}
		\nabla \cdot \mathbf{A}_{n} = 0\,, \qquad \forall\;  n \in \{1,\dots,N\}\,.
	\end{equation*}
This implies that \eqref{eq: HalN} can be equivalently expressed on $ C^{\infty}_{\mathrm{c}}(\mathbb{R}^2\!\setminus \!\lf\{ \xv_1, \dots, \xv_N \ri\}) $ as
	\begin{equation*}
		H_{N} = - \,\Delta + 2 \left(\mbox{$\sum_{n \,=\, 1}^{N}$}\, \mathbf{A}_{n}\right) \cdot (-i\nabla) + \left(\mbox{$\sum_{n \,=\, 1}^{N}$}\, \mathbf{A}_{n} \right)^{\!2},
	\end{equation*}
where the order of the factors in the second term on the r.h.s. is immaterial.

	\begin{remark}[Fluxes' intensities]
	\label{rem: an1}
	\mbox{}	\\
	The condition $\alpha_n \in (0,1)$ for any $n \in \{1,\dots,N\}$, rather than $\alpha_n \in \mathbb{R}$, actually entails no loss of generality, since it can always be realized via a unitary transformation. To prove this claim, assume $\alpha_n = a_n + \tilde{\alpha}_{n}$ with $a_n \in \mathbb{Z}$ and $\tilde{\alpha}_{n} \in (0,1)$, and consider the unitary maps
		\begin{equation*}
			U_n : L^2(\mathbb{R}^2) \to L^2(\mathbb{R}^2)\,, \qquad
			(U_n \psi)(\xv) =  e^{i a_n \arg(\xv - \xv_n)} \, \psi(\xv)\,,
		\end{equation*}
	for $ n \in \{1, \ldots, N\} $. Notice that $U_n \psi$ is indeed a {\it single-valued} function in $L^2(\mathbb{R}^2)$, given that  $e^{2\pi i a_n } = 1$. By composition, we proceed to define the unitary operator $  U := U_1 \cdots\, U_N $. Using the basic identity $\nabla \arg(\xv - \xv_n) = \frac{(\mathbf{x} - \mathbf{x}_{n})^{\perp}}{|\mathbf{x} - \mathbf{x}_{n}|^2} $, one easily gets
		\begin{equation*}
			U \,H_{N}\, U^{-1} = \left(- \,i \nabla + \sum_{n \,=\, 1}^{N}\, (\alpha_n - a_n)\,{(\mathbf{x} - \mathbf{x}_{n})^{\perp} \over |\mathbf{x} - \mathbf{x}_{n}|^2} \right)^{\!2},
		\end{equation*}
	which accounts for the above statement. 
	Let us also point out that the configuration with flux intensities $(\alpha_n)_{n \,\in\, \{1, \ldots, N\}}$ can be mapped to that with opposite parameters $(-\alpha_n)_{n \,\in\, \{1, \ldots, N\}}$ exploiting an obvious conjugation symmetry.
	\end{remark}
	
	Concerning the fluxes' distribution we only assume that their positions are distinct, {\it i.e.}, that there exists an $r_{*} $ fulfilling
	\begin{equation*}
		0 < r_{*} <\tfrac{1}{2} \disp\min_{m,n \,\in\, \{1,\,\dots\,,N\}}  \lf|\mathbf{x}_{m} - \mathbf{x}_{n}\ri| .
	\end{equation*}
	Accordingly, we can consider a partition of unity given by a family of $C^{2}$ functions $(\xi_{n})_{n \,\in\, \{0,1, \ldots, N\}} : \mathbb{R}^{2} \to [0,1]$ such that:
	\begin{gather}
		 \supp \xi_{n} \subset B_{r_{*}}(\mathbf{x}_{n}) \quad \mbox{and} \quad \xi_{n}\big|_{B_{r_{*}/2}(\mathbf{x}_{n})} \equiv 1\,, \qquad \forall\; n \in \{1,\dots,N\}\,, \label{eq: xinsupp}\\
		 \mbox{$\sum_{n \,=\, 0}^{N}$}\, \xi_{n}^2 = 1\,, \label{eq: sumxin1}
	\end{gather}
	where we denoted by $B_{\varrho}(\mathbf{x}) $  the open disc of radius $ \varrho > 0$ and center $\mathbf{x} $.
The above assumptions ensure, in particular, that
	\begin{gather}
		\supp \xi_{m} \cap\, \supp \xi_{n} = \varnothing\,, \qquad \forall\; m \neq n \,; \label{eq: ximxin}\\
		\mbox{$\exists\, R > 0$\quad s.t.\quad $\xi_0\big|_{\mathbb{R}^2 \setminus B_{R}(\mathbf{0})} = 1$}\,. \label{eq: xi01}
	\end{gather}
	In view of the above positions, we further introduce the notations, for $ n \in \lf\{1, \dots, N \ri\} $,
	\begin{equation*}
		\mathbf{S}_{n} := \mbox{$\sum_{m \,\neq\, n}$}\,\mathbf{A}_{m}\,, \qquad \mathbf{S}_{0} := \mbox{$\sum_{m \,=\, 1}^{N}$}\,\mathbf{A}_{m}\,, 	
	\end{equation*}
		to denote the regular part of the magnetic potential in a neighborhood of the $n$-th flux. Furthermore, we set $ \check{\mathbf{S}}_{n}(\mathbf{x}) := \mathbf{S}_{n}(\mathbf{x}) - \mathbf{S}_{n}(\mathbf{x}_{n}) $, so that
	\begin{equation}\label{eq: checkSnLip}
		\lf|\check{\mathbf{S}}_{n}(\mathbf{x})\ri| \leqslant c\,|\mathbf{x}-\mathbf{x}_n|\,, \qquad
		 \forall\; n \!\in \!\{1, \dots, N\},\;\, \forall\; \mathbf{x} \in \supp \xi_{n}\,.
	\end{equation}

\subsection{The Friedrichs extension}
A distinguished self-adjoint realization of the operator \eqref{eq: HalN} is the Friedrichs one. In this connection, let us consider the quadratic form associated to the positive operator $ H_N $, {\it i.e.},
	\begin{equation}\label{eq: QalN}
		Q_{N}[\psi] 
		:= \left\| \left(- i \nabla + \mbox{$\sum_{n \,=\, 1}^{N}$} \,\mathbf{A}_{n} \right) \psi\right\|_{2}^{\,2} ,
	\end{equation}
which is well defined at least on $C^{\infty}_{\mathrm{c}}(\mathbb{R}^2 \!\setminus\! \{\mathbf{x}_{1},\dots,\mathbf{x}_{N}\})$. Using the natural norm $\|\psi\|_{N}^2 := \|\psi\|_{2}^{2} + Q_{N}[\psi]$, we can identify the quadratic form associated to the Friedrichs extension of ${H}_{N}$, namely,
	\begin{equation}\label{eq: QFr}
		Q_{N}^{(\mathrm{F})}[\psi] = Q_{N}[\psi]\,,	\qquad		\dom\big[Q_{N}^{(\mathrm{F})}\big] = \overline{C^{\infty}_{\mathrm{c}}(\mathbb{R}^2 \!\setminus \!\{\mathbf{x}_{1},\dots,\mathbf{x}_{N}\})}^{\;\|\,\cdot\,\|_{N}} .
	\end{equation}
	
	For later purposes, let us denoted by $ \avg{f}_n: (0,+\infty) \to \mathbb{C} $ the angular average around $\mathbf{x}_{n}$ of any function $ f: \mathbb{R}^2 \to \mathbb{C} $, {\it i.e.}, 
	\begin{equation*}
		\avg{f}_{n}\!(r_n) : = \frac{1}{2\pi r_n} \int_{\partial B_{r_n}(\mathbf{x}_n)}\hspace{-0.3cm} \diff \Sigma_n \: f = \frac{1}{2 \pi} \int_0^{2\pi} \diff \vartheta_n \: f(r_n,\theta_n)\,,
	\end{equation*}
	where $ r_n : = \lf| \xv - \xv_n \ri| $ and $ \vartheta_n := \mathrm{arg}(\xv - \xv_n) $ are local polar coordinates around $ \xv_n $ and we have committed a little abuse of notation setting $ f(x_n,\vartheta_n) := f(\xv(x_n,\vartheta_n)) $. As we are going to see, exactly as in the single-flux case (see \cite[Proposition 1.1]{CF21}), one of the key properties of the Friedrichs extension is the vanishing of the angular average around each flux of any function $ \Psi $ in its domain. Namely, no singular behavior is present in $ \Psi $ and the terms $ \nabla \Psi $ and $ \mathbf{A}_n  \Psi $ are separately in $ L^2(\mathbb{R}^2) $.
	
\begin{proposition}[Friedrichs extension]	
		\label{prop:QF}
		\mbox{}	\\ 
		Let $ N \in \N $ and, for all $n = 1, \ldots, N$, let $\alpha_n \in (0,1)$. Then,
		\begin{enumerate}[i)]
			\item The quadratic form $Q_{N}^{(\mathrm{F})}$ defined in \eqref{eq: QFr} is closed and non-negative. Its domain is given by
				\begin{equation}\label{eq: domQF}
					\dom \big[Q_{N}^{(\mathrm{F})}\big] = \lf\{\psi \in H^1(\mathbb{R}^2)\; \big|\;\mathbf{A}_{n} \psi \in L^2(\mathbb{R}^2)\,,\;\, \forall\; n \!\in\! \lf\{  1,\dots,N \ri\} \ri\}.
				\end{equation}
				Moreover, for any $\psi \in \dom\big[Q_{N}^{(\mathrm{F})}\big]$ and for all $n \in \{1,\dots,N\}$,
				\begin{equation}\label{eq: limdomQF}
				\lim_{r_n \to\, 0^+} \avg{|\psi|^2}_{\!n} = 0 \,, 	\qquad
				\lim_{r_n \to\, 0^+} r_n^2 \avg{|\partial_{r_n}\psi|^2}_{\!n} = 0\,.
				\end{equation}
			\item The self-adjoint operator $H_{N}^{(\mathrm{F})}$ associated to $Q_{N}^{(\mathrm{F})}$ acts as $H_{N}$ on the domain
				\begin{equation}\label{eq: domHF}
					\dom\big(H_{N}^{(\mathrm{F})}\big) = \lf\{\psi \in \dom \big[ Q_{N}^{(\mathrm{F})} \big] \;\Big|\; H_{N}\psi \in L^2(\mathbb{R}^2) \ri\} .
				\end{equation}
		\end{enumerate}
\end{proposition}

\subsection{Heuristic derivation of the perturbed quadratic forms}

In order to derive the explicit expressions of the quadratic forms extending \eqref{eq: QFr}, one can perform the heuristic computation described hereafter (see \cite[Eqs. (1.24)-(1.27)]{CF21} for a comparison with the single-flux case). The key idea is that, as for a single flux, the regular wave functions in the Friedrichs realization's domain can be perturbed by adding defect functions with prescribed singularities at the fluxes. 

For $n \in \{1,\dots,N\}$ and $\lambda > 0$, let us then consider the two defect functions $G^{(\ell_n)}_{\lambda,n} \in L^2(\mathbb{R}^2)$, $\ell_n \in \{0,-1\}$, associated to the configuration with a single AB flux of intensity $\alpha_n \in (0,1)$, placed at the point $\mathbf{x}_n \in \mathbb{R}^2$. Here $ \ell_n $ singles out the angular momentum subspace ($s$ and $p${\it -waves} only). The functions $G^{(\ell_n)}_{\lambda,n}$ are identified as the unique solutions of the deficiency equation
	\begin{equation}\label{eq: greeneq}
		\lf((- i \nabla + \mathbf{A}_{n})^{2} + \lambda^2\ri) G^{(\ell_n)}_{\lambda,n} = 0\,, \qquad \mbox{in\; $\mathbb{R}^2 \setminus \{\mathbf{x}_{n}\}$}\,,
	\end{equation}
	{\it i.e.}, explicitly,
	\begin{equation}\label{eq: G2exp}
		G^{(\ell_n)}_{\lambda,n}(r_n,\theta_n) = \lambda^{|\ell_n + \alpha_n|}\, K_{|\ell_n + \alpha_n|}(\lambda\,r_n)\, \frac{e^{i\, \ell_n \theta_n}}{\sqrt{2\pi}}\,,\qquad \mbox{for\, $\ell_n \in \{0,-1\}$}\,, 
	\end{equation}
where $K_{\nu}$ is the modified Bessel function of second kind, {\it a.k.a.} Macdonald function. Let us stress that $G^{(\ell_n)}_{\lambda,n}$ is square-integrable \cite[Eq.\,6.521.3]{GR07} with $L^2(\mathbb{R}^2)$ norm
	\begin{equation*}
	\lf\|G^{(\ell_n)}_{\lambda,n}\ri\|_{2}^{2} = {\pi\, |\ell_n + \alpha_n| \over 2\sin(\pi\,\alpha_n)}\; \lambda^{2\,|\ell_n + \alpha_n|-2} \,,
	\end{equation*}
	and has the following asymptotics in a neighborhood of $ \xv_n $:
	\bml{\label{eq: G2asy0}
		G^{(\ell_n)}_{\lambda,n}(r_n,\theta_n) = \left[{\Gamma\big(|\ell_n + \alpha_n|\big) \over 2^{1 - |\ell_n + \alpha_n|}}\, {1 \over r_n^{|\ell_n + \alpha_n|}} + {\Gamma\left(-|\ell_n + \alpha_n|\right) \over 2^{1 + |\ell_n + \alpha_n|}}\,\lambda^{2 |\ell_n + \alpha_n|}\,r_n^{|\ell_n + \alpha_n|} \ri.	\\
		\lf. +\; \mathcal{O}\left(r_n^{2-|\ell_n + \alpha_n|}\right)\right]\! {e^{i\, \ell_n \theta_n} \over \sqrt{2\pi}} \,.
	}
Notice the singular term $\sim\, r_n^{-|\ell_n + \alpha_n|} $ with coefficient independent of $\lambda$, which ensures that $ G^{(\ell_n)}_{\lambda,n} \notin \dom \big[Q_{N}^{(\mathrm{F})}\big]$  since it cannot fulfill the asymptotic conditions in \eqref{eq: limdomQF}.

In order to set the singular behavior around $ \xv_n $ of the wave functions, we perturb functions $ \phi_{\lambda} \in \dom \big[Q_{N}^{(\mathrm{F})}\big] $ by adding a term of the form
\begin{equation*}
		\lf( \chi_{\lambda} \mathbf{q} \ri)(\xv) := \sum_{n \,=\, 1}^{N} \, e^{-i \mathbf{S}_{n}(\mathbf{x}_n) \cdot (\mathbf{x} - \mathbf{x}_n)}\, \xi_{n}(\xv)\, \sum_{\ell_n \,\in\, \{0,-1\}} \, q^{(\ell_n)}_{n}\, G^{(\ell_n)}_{\lambda,n}(\xv-\xv_n)\,,	
\end{equation*}
where $ \mathbf{q} = \big(q_1^{(0)}, q_1^{(-1)}, q_2^{(0)},\, \ldots, q_N^{(0)}, q_N^{(-1)}\big) \in \mathbb{C}^{2N} $ are free coefficients, the functions $ \xi_n $, $ n = 1, \ldots, N $, belong to a partition of unity as introduced in \eqref{eq: xinsupp} and \eqref{eq: sumxin1}, and the phase factor $ e^{-i \mathbf{S}_{n}(\mathbf{x}_n) \cdot (\mathbf{x} - \mathbf{x}_n)} $ has been inserted for later convenience. Assuming for the sake of simplicity that $ \phi_{\lambda} \in C^{\infty}_{\mathrm{c}}(\mathbb{R}^2\setminus\{\mathbf{x}_1,\dots,\mathbf{x}_{N}\} $ and formally evaluating the expectation of $ H_N $ on $ \psi = \phi_{\lambda} + \chi_{\lambda} \mathbf{q} $, we get
\bml{
	\label{eq: expectation H_N}
	\meanlrlr{\psi}{H_{N}^{(\mathrm{F})}}{\psi} = \meanlrlr{\phi_{\lambda}}{H_{N}^{(\mathrm{F})}}{\phi_{\lambda}}
		+ 2 \tx\sum_{n,\ell_n} \Re \lf[ q^{(\ell_n)}_{n} \meanlrlr{\phi_{\lambda}}{H_{N}^{(\mathrm{F})}}{e^{-i \mathbf{S}_{n}(\mathbf{x}_n) \cdot (\mathbf{x} - \mathbf{x}_n)}\, \xi_{n} G^{(\ell_n)}_{\lambda,n}} \ri] \\
		+ \tx\sum_{n,\ell_n,\ell_n^{\prime}} \big({q^{(\ell_n)}_{n}}\big)^* q^{(\ell_n^{\prime})}_{n} \meanlrlr{ e^{-i \mathbf{S}_{n}(\mathbf{x}_n) \cdot (\mathbf{x} - \mathbf{x}_n)}\, \xi_{n} G^{(\ell_n)}_{\lambda,n}}{H_{N}^{(\mathrm{F})}}{e^{-i \mathbf{S}_{n}(\mathbf{x}_n) \cdot (\mathbf{x} - \mathbf{x}_n)}\,\xi_{n} G^{(\ell_n^{\prime})}_{\lambda,n}}\,.
}
Note the absence of off-diagonal terms in the last sum, thanks to the disjoint supports of the functions $ \xi_n $. For $ \mathbf{x} \neq \mathbf{x}_{n}$ we have
	\begin{align*}
		& H_{N}^{(\mathrm{F})}\! \lf( e^{-i \mathbf{S}_{n}(\mathbf{x}_n) \cdot (\mathbf{x} - \mathbf{x}_n)}\, \xi_{n} G^{(\ell_n)}_{\lambda,n} \ri) = e^{-i \mathbf{S}_{n}(\mathbf{x}_n) \cdot (\mathbf{x} - \mathbf{x}_n)}\, \lf(-i \nabla + \mathbf{A}_n + \mathbf{S}_n- \mathbf{S}_{n}(\mathbf{x}_n)\ri)^2\lf( \xi_{n} G^{(\ell_n)}_{\lambda,n} \ri) \\
		& = e^{-i \mathbf{S}_{n}(\mathbf{x}_n) \cdot (\mathbf{x} - \mathbf{x}_n)} \lf[  
			2\!\left(\check{\mathbf{S}}_n \xi_{n}\! -\! i \nabla \xi_{n}\right) \!\cdot\! \lf(-i \nabla \!+\! \mathbf{A}_n\ri) G^{(\ell_n)}_{\lambda,n}
			+ \Big( \big(\check{\mathbf{S}}_n^2\! - \lambda^2\big) \xi_{n} + 2 \check{\mathbf{S}}_n \!\cdot\! (-i \nabla \xi_{n}) - \Delta \xi_{n} \Big) G^{(\ell_n)}_{\lambda,n}
			\ri],
\end{align*}
and, since the support of $ \phi_{\lambda} $ does not comprises $\mathbf{x}_{n}$, we can integrate by parts without getting any boundary term (we shall return on this point in the proof of \cref{cor:domHbeta}), so obtaining
\bmln{
	\meanlrlr{\phi_{\lambda}}{H_{N}^{(\mathrm{F})}}{e^{-i \mathbf{S}_{n}(\mathbf{x}_n) \cdot (\mathbf{x} - \mathbf{x}_n)}\,\xi_{n} G^{(\ell_n)}_{\lambda,n}} 
	= 2 \braketr{ \lf(-i \nabla \!+\! \mathbf{A}_{n}\ri)\phi_{\lambda}}{\,e^{-i \mathbf{S}_{n}(\mathbf{x}_n) \cdot (\mathbf{x} - \mathbf{x}_n)} \big(\check{\mathbf{S}}_{n}\,\xi_{n} - i \nabla \xi_{n} \big) G^{(\ell_n)}_{\lambda,n}}\\
	+ \braketr{ \phi_{\lambda}}{e^{-i \mathbf{S}_{n}(\mathbf{x}_n) \cdot (\mathbf{x} - \mathbf{x}_n)}\lf( \big(\check{\mathbf{S}}_{n}^2 - \lambda^2\big)\, \xi_{n} + 2 \mathbf{S}_n(\mathbf{x}_n) \cdot (\check{\mathbf{S}}_n \xi_{n} - i \nabla \xi_{n}) + \Delta \xi_{n} \ri)\! G^{(\ell_n)}_{\lambda,n}}.
}
Similarly, we deduce
\bmln{
		\meanlrlr{ e^{-i \mathbf{S}_{n}(\mathbf{x}_n) \cdot (\mathbf{x} - \mathbf{x}_n)}\,\xi_{n} G^{(\ell_n)}_{\lambda,n}}{H_{N}^{(\mathrm{F})}}{e^{-i \mathbf{S}_{n}(\mathbf{x}_n) \cdot (\mathbf{x} - \mathbf{x}_n)}\,\xi_{n} G^{(\ell_n^{\prime})}_{\lambda,n}}
		= 2 \braketr{\xi_{n} G^{(\ell_n)}_{\lambda,n}}{ \left(- i \nabla \xi_{n}\right) \!\cdot\! \lf(-i \nabla \!+\! \mathbf{A}_n\ri) G^{(\ell_n^{\prime})}_{\lambda,n} }
		\\
		+ \braketr{\xi_{n} G^{(\ell_n)}_{\lambda,n}}{ \Big(\big(\check{\mathbf{S}}_n^2\! - \lambda^2\big) \xi_{n} + 2 \big(\check{\mathbf{S}}_n \!\cdot\! \mathbf{A}_n\big)\xi_{n} - \Delta \xi_{n} \Big) G^{(\ell_n^{\prime})}_{\lambda,n} } + 2  \braketr{\xi_{n} G^{(\ell_n)}_{\lambda,n}}{ \,\check{\mathbf{S}}_n \!\cdot\! (-i \nabla) \big(\xi_{n}  G^{(\ell_n^{\prime})}_{\lambda,n}\big)} , 
}
so that, symmetrizing the last term of \eqref{eq: expectation H_N} and integrating by parts, we get
\bmln{
		\tx\sum_{n,\ell_n,\ell_n^{\prime}} \big({q^{(\ell_n)}_{n}}\big)^* q^{(\ell_n^{\prime})}_{n}\meanlrlr{ e^{-i \mathbf{S}_{n}(\mathbf{x}_n) \cdot (\mathbf{x} - \mathbf{x}_n)}\,\xi_{n} G^{(\ell_n)}_{\lambda,n}}{H_{N}^{(\mathrm{F})}}{e^{-i \mathbf{S}_{n}(\mathbf{x}_n) \cdot (\mathbf{x} - \mathbf{x}_n)}\,\xi_{n} G^{(\ell_n^{\prime})}_{\lambda,n}} \\
		= \tx\sum_{n,\ell_n,\ell_n^{\prime}} \big({q^{(\ell_n)}_{n}}\big)^* q^{(\ell_n^{\prime})}_{n} \bigg[ \braketr{G^{(\ell_n)}_{\lambda,n} }{(\nabla \xi_{n})^2 G^{(\ell_n^{\prime})}_{\lambda,n}} + \braketr{\xi_{n} G^{(\ell_n)}_{\lambda,n}}{ \Big(\big(\check{\mathbf{S}}_n^2\! - \lambda^2\big) \xi_{n} + 2 \big(\check{\mathbf{S}}_n \!\cdot\! \mathbf{A}_n\big)\xi_{n} \Big) G^{(\ell_n^{\prime})}_{\lambda,n} }
		\\
			+ 2 \braketr{\xi_{n} G^{(\ell_n)}_{\lambda,n}}{ \,\check{\mathbf{S}}_n \!\cdot\! (-i \nabla) \big(\xi_{n}  G^{(\ell_n^{\prime})}_{\lambda,n}\big)} \bigg] \,.
}
Exploiting the identity
\bmln{
	\lf\| \psi \ri\|_{2}^2 = \lf\| \phi_{\lambda} \ri\|_{2}^2 + 2 \tx\sum_{n,\ell_n} \Re \Big[ q^{(\ell_n)}_{n} \braketr{\phi_{\lambda}}{ \,e^{-i \mathbf{S}_{n}(\mathbf{x}_n) \cdot (\mathbf{x} - \mathbf{x}_n)}\,\xi_{n} G^{(\ell_n)}_{\lambda,n}} \Big] \\
	+ \tx\sum_{n,\ell_n,\ell_n^{\prime}} \big({q^{(\ell_n)}_{n}}\big)^* q^{(\ell_n^{\prime})}_{n} \braketr{\xi_{n} G^{(\ell_n)}_{\lambda,n}}{\xi_{n} G^{(\ell_n^{\prime})}_{\lambda,n}},
}
we are finally led to consider the expression
\bml{\label{eq: Qbeta}
		Q^{(B)}_{N}[\psi] := Q^{(\mathrm{F})}_{N}[\phi_{\lambda}] - \lambda^2 \lf\| \psi \ri\|_{2}^2 + \lambda^2 \lf\| \phi_{\lambda} \ri\|_{2}^2 \\
		+ 2 \sum_{n = 1}^N \sum_{\ell_n \in \{ 0, -1 \}} \Re \lf[ q^{(\ell_n)}_{n} \lf( 2 \braketr{ \lf(-i \nabla \!+\! \mathbf{A}_{n}\ri) \phi_{\lambda}}{\zeta_{1,n}\, G^{(\ell_n)}_{\lambda,n}} + \braketr{ \phi_{\lambda}}{\zeta_{2,n}\, G^{(\ell_n)}_{\lambda,n}} \ri) \ri] \\
		+  \sum_{m, n = 1}^N \sum_{\ell_m, \ell_n^{\prime} \in \{ 0, -1 \}}  \big({q^{(\ell_m)}_{m}}\big)^* q^{(\ell_n^{\prime})}_{n} \left[B^{(\ell_m \ell_n^{\prime})}_{m\,n} + \delta_{mn} \lf( \tfrac{\pi\,\lambda^{2|\ell_n + \alpha_n|}}{2 \sin(\pi \alpha_n)}\,\delta_{\ell_n \ell_n^{\prime}} + \Xi^{(\ell_n \ell_n^{\prime})}_{n}(\lambda) \ri) \right] ,
}
where
\begin{gather}
	\zeta_{1,n}(\xv)  : =  e^{-i \mathbf{S}_{n}(\mathbf{x}_n) \cdot (\mathbf{x} - \mathbf{x}_n)} \big(\check{\mathbf{S}}_{n}(\xv)\,\xi_{n}(\xv) - i \lf(\nabla \xi_{n}\ri)(\xv) \big)\,, 
	\\
	\zeta_{2,n}(\xv)  : =  e^{-i \mathbf{S}_{n}(\mathbf{x}_n) \cdot (\mathbf{x} - \mathbf{x}_n)}\lf[ \check{\mathbf{S}}_{n}^2(\xv)\, \xi_{n}(\xv) + 2\, \mathbf{S}_n(\mathbf{x}_n) \cdot \lf(\check{\mathbf{S}}_n(\xv) \xi_{n}(\xv) - i \lf(\nabla \xi_{n}\ri) (\xv) \ri) + \lf(\Delta \xi_{n}\ri)(\xv) \ri],	
	\\
	\Xi^{(\ell_n \ell_n^{\prime})}_{n}(\lambda)  :=   \braketr{G^{(\ell_n)}_{\lambda,n}}{ \lf[\big(\check{\mathbf{S}}_n^2 + 2 \check{\mathbf{S}}_n \!\cdot\! \mathbf{A}_n \big)\xi_{n}^2 + (\nabla \xi_{n})^2 \ri] G^{(\ell_n^{\prime})}_{\lambda,n} }
		+ 2 \braketr{\xi_{n} G^{(\ell_n)}_{\lambda,n}}{ \,\check{\mathbf{S}}_n \!\cdot\! (-i \nabla) \big(\xi_{n}  G^{(\ell_n^{\prime})}_{\lambda,n}\big)\!} ,
	\label{eq: defXi}
\end{gather} 
and we have introduced the Hermitian matrix  
	\begin{equation*}
		B : =  \lf(B_{m\,n}^{(\ell_m \ell'_n)}\ri)_{m,n \in \{1, \ldots, N\}; \, \ell_m, \ell'_n \in \{0,-1\}} \in M_{2N,\,\mathrm{Herm}}(\mathbb{C})\,,
	\end{equation*}
	labeling the form. Notice that, thanks to the properties of the cut-off functions $ \xi_n $ and \eqref{eq: checkSnLip}, we have $ \zeta_{1,n}, \zeta_{2,n},  \check{\mathbf{S}}_n \cdot \mathbf{A}_n \xi_n \in L^{\infty}(\R^2) $. Therefore, all the scalar products appearing in the quadratic form are well-posed: in particular, for the last term in $\Xi^{(\ell_n \ell_n^{\prime})}_{n}(\lambda) $, one has to use that $ \check{\mathbf{S}}_n $ linearly vanishes around $ \xv_n$ to compensate the extra singularity due to the gradient. Since it can be checked by direct inspection that $ \Xi^{(\ell_n \ell_n^{\prime})}_{n}(\lambda) = \big({\Xi^{(\ell_n^{\prime} \ell_n)}_{n}(\lambda)}\big)^{*} $ and $ B $ is Hermitian, we also infer that the form is real.

\subsection{Main results}

Our first result is about the quadratic forms $ Q^{(B)}_{N} $ defined in \eqref{eq: Qbeta} on their natural domain of definition, {\it i.e.},
	\bml{
 		\dom\big[Q^{(B)}_{N}\big] := \bigg\{\psi \!\in\! L^2(\mathbb{R}^2)\,\bigg|\, \psi = \phi_{\lambda} + \sum_{n \,=\, 1}^{N}\, e^{-i \mathbf{S}_{n}(\mathbf{x}_n) \cdot (\mathbf{x} - \mathbf{x}_n)}\, \xi_{n}\, \sum_{\ell_n \,\in\, \{0,-1\}}\, q^{(\ell_n)}_{n}\, G^{(\ell_n)}_{\lambda,n} \,,\,  \\
 			 \phi_{\lambda} \!\in\! \dom\big[Q_{N}^{(\mathrm{F})}\big],\;\, q^{(\ell_n)}_{n} \!\in\! \mathbb{C},\;\, n = 1,\dots, N,\;\, \ell_n = 0,-1 \bigg\} \,. \label{eq: Qbetadom}
	}

	\begin{theorem}[Quadratic forms $ Q^{(B)}_{N} $]
		\label{thm: Qbeta}
		\mbox{}		\\
		Let $ N \in \N $ and, for all $n = 1,\dots,N$, let $\alpha_n \in (0,1)$ and $\xi_{n} : \mathbb{R}^2 \to [0,1]$ fulfil \eqref{eq: xinsupp}~-~\eqref{eq: xi01}. Then, for any Hermitian matrix $B \in  M_{2N,\,\mathrm{Herm}}(\mathbb{C})$, 
		 \ben[i)]
		 	\item the quadratic form $Q^{(B)}_{N}$ defined by \eqref{eq: Qbeta} is well-posed on the domain \eqref{eq: Qbetadom}; moreover, it is independent of $ \lambda > 0 $ and of the choice of $ (\xi_{n})_{n \,=\, 1,\,\dots\,,N}$;
		 	\item $Q^{(B)}_{N}$ is also closed and bounded from below on the same domain.
		 \een
	\end{theorem}

	Next we show that the self-adjoint operators associated to the forms $Q^{(B)}_{N}$ identify all self-adjoint realizations of $ H_N $. Besides the derivation of the domain and explicit action of such operators, including the boundary conditions satisfied by the functions contained therein, the major content of the result reported below is the fact that {\it all} the self-adjoint extensions are contained in the family. This is obtained indirectly through an alternative parametrization of the family via Kre{\u\i}n's theory (see also the subsequent \cref{rem: krein}).

	\begin{theorem}[Self-adjoint extensions $ H_{N}^{(B)} $]
		\label{cor:domHbeta}
		\mbox{}	\\
		Under the same assumptions of \cref{thm: Qbeta}, for any  Hermitian matrix $ B \in M_{2N,\,\mathrm{Herm}}(\mathbb{C})$ the operator $ H_{N}^{(B)} $ associated to the quadratic form $Q^{(B)}_{N}$ has domain 
		\bml{
			\label{domHbe1}
			\dom \big(H_{N}^{(B)}\big) 
			= \lf\{ \psi = \phi_{\lambda} + \mbox{$\sum_{n,\ell_n}$}\, q^{(\ell_n)}_{n}\, e^{-i \mathbf{S}_{n}(\mathbf{x}_n) \cdot (\mathbf{x} - \mathbf{x}_n)}\, \xi_{n}\, G_{\lambda,n}^{(\ell_n)} \;\Big|\;
			\phi_{\lambda} \in \dom\big(H_{N}^{(\mathrm{F})}\big)\,,  \ri.	\\
			\lf.  \big[ \lf(B +  L  \ri) \mathbf{q} \big]_{m, \ell_m}\!
				= \lim_{r \to 0^{+}} \frac{ \pi\,2^{|\ell_m + \alpha_m|} \,\Gamma\big(|\ell_m + \alpha_m|\big)}{r^{|\ell_m + \alpha_m|}}\,\avg{\big( |\ell_m + \alpha_m|\,\phi_{\lambda} + r\,\partial_{r} \phi_{\lambda}\big) \tfrac{e^{-i\, \ell_m \theta}}{\sqrt{2\pi}}}_{m},\, \forall\, m, \ell_m   \ri\} ,
			}	
			where $ L : = \lf(\tfrac{\pi\,\lambda^{2|\ell'_n + \alpha_n|}}{ 2\sin(\pi\alpha_n)}\,\delta_{mn}\,\delta_{\ell_m \ell_n^{\prime}} \ri)_{m,n \in \{1, \ldots, N\}; \, \ell_m, \ell'_n \in \{0,-1\}} \in M_{2N,\,\mathrm{Herm}}(\mathbb{C})$, and action
			\bmln{
				\lf( H_{N}^{(B)} \!+\! \lambda^2\ri)\! \psi  
				= \lf( H_{N}^{(\mathrm{F})} \!+\! \lambda^2\ri)\! \phi_{\lambda}
					+ \sum_{n,\ell_n} q^{(\ell_n)}_{n}\, e^{-i \mathbf{S}_{n}(\mathbf{x}_n) \cdot (\mathbf{x} - \mathbf{x}_n)} \Big[  2\big(\check{\mathbf{S}}_{n}\xi_{n}\! - i \nabla \xi_{n} \big) \lf(-i \nabla \!+\! \mathbf{A}_{n}\ri)G^{(\ell_n)}_{\lambda,n} \\
					+ \lf( \check{\mathbf{S}}_{n}^2 \xi_{n}\! + 2\check{\mathbf{S}}_{n} \!\cdot\! (-i \nabla \xi_{n}) - \Delta \xi_{n} \ri)\! G^{(\ell_n)}_{\lambda,n} \Big]\,.
			}
			Moreover the $(2N)^2$-real-parameters family $H_{N}^{(B)}$, $B \in M_{2N,\,\mathrm{Herm}}(\mathbb{C}) \cup \lf\{ \infty \ri\}$, exhausts all possible self-adjoint realizations of the operator $ {H}_{N} $.
	\end{theorem}

\begin{remark}[Friedrichs extension]
	\label{remark: HFbc}
	\mbox{}	\\
	The Friedrichs Hamiltonian is formally recovered for ``$B = \infty$'' and from now on we may use the notation $ B = \infty $ to identify the Friedrichs extension. Indeed, when all components of $  B$ diverge, all the charge parameters $q_{n}^{(\ell_n)}$ are set equal to zero. In this case, the boundary conditions in \eqref{domHbe1} for $m \in \{1,\dots,N\}$, $\ell_m \in \{0,-1\}$, read
	\begin{equation*}
		\avg{\Big( |\ell_m + \alpha_m|\, \phi_{\lambda} + r\, \partial_r \phi_{\lambda}\Big) \tfrac{e^{-i\, \ell_m \theta}} {\sqrt{2\pi}}}_m = \mathcal{O}\lf(r_m^{|\ell_m + \alpha_m|}\ri), \qquad \mbox{for\, $r_n \to 0^{+}$}.
	\end{equation*}
	This asymptotic behavior is apparently missing in the characterization of the Friedrichs domain \eqref{eq: domQF}, but it is in fact encoded in the requirement $H_{N} \phi_\lambda \in L^2(\mathbb{R}^2)$ (see also \cite[\S 3.1]{CO18}).
\end{remark}

\begin{remark}[Kre{\u\i}n representation]
	\label{rem: krein}
	\mbox{}	\\
	As anticipated, a key point in the proof of \cref{cor:domHbeta} is an alternative representation of the self-adjoint realizations of $ H_N $ via a straightforward application of Kre{\u\i}n theory. More precisely, one can show (see \cref{thm: extRF}) that all self-adjoint extensions of $ H_N$ can be written in the following form, for $ z \in \mathbb{C} \setminus \mathbb{R} $ and $ \Theta \in M_{2N,\,\mathrm{Herm}}(\mathbb{C})  \cup \lf\{ \infty \ri\} $:
	\begin{gather}
			\dom\big(H_{N}^{(\Theta)}\big) = \Big\{ \psi \in L^2(\mathbb{R}^2)\;\Big|\; \psi = \varphi_{z} + \mathcal{G}(z) \mathbf{q},\; \varphi_{z} \in \dom\big(H_{N}^{(\mathrm{F})}\big),\; \mathbf{q} \!\in \mathbb{C}^{2N}, \; \bm{\tau} \varphi_z = \big[\Theta + \Lambda(z)\big] \mathbf{q} \Big\}\,, \nonumber
	\\
		\big(H_{N}^{(\Theta)} - z\big) \psi = \big(H_{N}^{(\mathrm{F})} - z\big) \varphi_z\,. \label{eq:HNTheta}
	\end{gather}
	Here $ \bm{\tau} $ stands for the trace map $ \bm{\tau}  :=  \bigoplus_{n = 1, \ldots, N; \ell_n \in \{0,-1\}} \tau_{n}^{(\ell_n)} : \dom \big(H_{N}^{(\mathrm{F})}\big) \to \mathbb{C}^{2N} $, where ({\it cf.} \eqref{domHbe1})
	\beq
		\tau_{n}^{(\ell_n)} \psi :=  \lim_{r \to 0^{+}} \tfrac{\pi\,2^{|\ell_n + \alpha_n|} \,\Gamma\big(|\ell_n + \alpha_n|\big)}{r^{|\ell_n + \alpha_n|}}\,\avg{ \Big( |\ell_n + \alpha_n|\,\psi + r\,\partial_{r} \psi\Big) \tfrac{e^{-i\, \ell_n \theta}} {\sqrt{2\pi}}}_{n}. \label{eq: deftau}
	\eeq
	The associated single layer operator is
	\begin{equation*}
		\mathcal{G}(z) := \big(\breve{\mathcal{G}}(\bar{z}) \big)^{*} : \mathbb{C}^{2N} \to L^2(\mathbb{R}^2)\,,
	\end{equation*}
	with  $ \breve{\mathcal{G}}(z) :=  \bm{\tau} \, ( H^{(\mathrm{F})}_N - z )^{-1} : L^2(\mathbb{R}^2) \to \mathbb{C}^{2N}$. 
	Moreover, fixing arbitrarily $z_0 \in \mathbb{C} \setminus [0,+\infty)$, we have set
	\begin{equation}\label{eq: defLambda}
		\Lambda(z) := \bm{\tau} \lf( \tfrac{1}{2} \lf( \mathcal{G}(z_0) + \mathcal{G}(\bar{z}_0) \ri) - \mathcal{G}(z) \ri) : \mathbb{C}^{2N} \to \mathbb{C}^{2N} .
	\end{equation}
	The two representations of the self-adjoint extensions are completely equivalent, {\it i.e.}, there is a one-to-one correspondence between the two families, which we denote by $ \Theta(B) $, $ \Theta: M_{2N,\,\mathrm{Herm}}(\mathbb{C}) \mapsto M_{2N,\,\mathrm{Herm}}(\mathbb{C}) $ (see \cref{prop:BTheta} for the explicit form of such a change of parametrization).
\end{remark}	

We now focus on the main spectral and scattering properties of the operators $  H_{N}^{(B)} $. The crucial ingredient in this framework is an explicit expression of the resolvent operator of any self-adjoint realization. We start by observing that the resolvent of the Friedrichs extension admits a convenient representation (see \cref{lemma: Tn}):
			\begin{equation}
				R^{(\mathrm{F})}_{N}(z) : = \big( H^{(\mathrm{F})}_N - z \big)^{-1}  = \sum_{n \,=\, 0}^{N}\, e^{- i \mathbf{S}_{n}(\mathbf{x}_n) \cdot (\mathbf{x} - \mathbf{x}_n)}\xi_{n}\, R^{(\mathrm{F})}_{n}(z)\, \xi_{n}\,e^{i \mathbf{S}_{n}(\mathbf{x}_n) \cdot (\mathbf{x} - \mathbf{x}_n)}\, \big[1 + T_N(z)\big]^{-1} , \label{eq: RNFried}	
			\end{equation}
			where, for any $n \in \{1,\dots,N\}$, $ R^{(\mathrm{F})}_{n}(z) : =  ( H^{(\mathrm{F})}_n - z )^{-1} $ and $ H^{(\mathrm{F})}_n = (-i \nabla \!+\! \mathbf{A}_{n})^2 $ stands for  the Friedrichs realization of the Schr\"odinger operator corresponding to a single AB flux of intensity $\alpha_n$, placed at $\mathbf{x}_{n}$. We recall that the integral kernel of $ R^{(\mathrm{F})}_{n}(z) $ can be expressed as (see \cite[Eq.\! (3.2)]{AT98} and \cite[Eqs. 10.27.6-7]{OLBC10})
	\begin{equation}\label{eq: RnAT}
		R^{(\mathrm{F})}_{n}\big(z; r_n,\theta_n; r'_n,\theta'_n\big) 
			= \sum_{\ell \,\in\, \mathbb{Z}} \tfrac{1}{2 \pi} \; I_{|\ell + \alpha|}\big(\!- i \sqrt{z} \, (r_n \wedge r'_n)\big)\,K_{|\ell + \alpha|}\big(\!- i \sqrt{z} \,(r_n \vee r'_n)\big) \; e^{i \ell (\theta_n - \theta'_n)} \,.
	\end{equation}
Here and in the sequel for any complex number $z \in \mathbb{C} \setminus \mathbb{R}^{+}$ we always consider the determination of the square root with $\Im \sqrt{z} > 0$.
By convention, we identify $R_{0}^{(F)}(z) \equiv R_{0}(z) := (-\Delta - z)^{-1} $ with the resolvent of the free Hamiltonian, whose integral kernel is (see \cite[Eqs. 10.27.6-7]{OLBC10})
	\begin{equation*}
		R_{0}(z; \mathbf{x}; \mathbf{x}') 
		= \tfrac{i}{4}\, H^{(1)}_{0}\big(\sqrt{z} \;|\mathbf{x} - \mathbf{x}'|\,\big)
		= \tfrac{1}{2 \pi}\, K_{0}\big(\!-i \sqrt{z} \,|\mathbf{x} - \mathbf{x}'|\big) \,.
	\end{equation*}
On the other hand, the operator $ T_N $ appearing in \eqref{eq: RNFried} is given by
\begin{align}
		T_N(z) &:= \sum_{n \,=\, 0}^{N}  e^{-i \mathbf{S}_{n}(\mathbf{x}_n) \cdot (\mathbf{x} - \mathbf{x}_n)} P_n\, R^{(\mathrm{F})}_{n}(z)\, \xi_{n}\,e^{i \mathbf{S}_{n}(\mathbf{x}_n) \cdot (\mathbf{x} - \mathbf{x}_n)}\,, \label{eq: TNdef}\\
		P_n &:= 2 \,\lf(\check{\mathbf{S}}_{n} \xi_{n} - i \nabla \xi_{n}\ri) \!\cdot\! (-i\nabla \!+\! \mathbf{A}_{n})
				+ \check{\mathbf{S}}_{n}^2 \xi_{n} + 2 \,\check{\mathbf{S}}_{n}\! \cdot ( - i \nabla \xi_{n})\! - \Delta \xi_{n}\,.
\end{align}

Combining the above representation of the Friedrichs resolvent with the Kre{\u\i}n formula for the resolvent of the self-adjoint extensions (recall \cref{rem: krein}), {\it i.e.},
		\begin{equation*}
		R_N^{(B)}(z) := \big( H^{(\Theta(B))}_N - z \big)^{-1} = R_N^{(\mathrm{F})}(z) + \mathcal{G}(z) \big[ \Theta(B) + \Lambda(z) \big]^{-1} \breve{\mathcal{G}}(z)\,,
		\end{equation*}
where the explicit form of the map $ \Theta(B) $ is given in \eqref{eq: TeB}, we are able to prove the following results.

\begin{proposition}[Spectral properties]
	\label{pro: spectrum}
	\mbox{}	\\
	Let $ N \in \N $ and, for all $n = 1,\dots,N$, let $\alpha_n \in (0,1)$. Then, for any Hermitian matrix $B \in M_{2N,\,\mathrm{Herm}}(\mathbb{C}) $, 
	\begin{equation*}
		\sigma_{\mathrm{ac}}\big(H_{N}^{(B)}\big) = \sigma_{\mathrm{ac}}(-\Delta) = [0,+\infty)\,.
	\end{equation*}
	Furthermore, $ \sigma_{\mathrm{disc}}\big(H_{N}^{(\mathrm{F})}\big) = \varnothing $, while $ \sigma_{\mathrm{disc}}\big(H_{N}^{(B)}\big) $ contains at most $ 2N $ negative eigenvalues. More precisely,
		\begin{equation*}
			-\lambda^2 \in \sigma_{\mathrm{disc}}\big(H_{N}^{(B)}\big) \quad (\lambda > 0) \quad \Longleftrightarrow \quad \ker \lf[ \Theta(B) + \Lambda(-\lambda^2) \ri] \neq \varnothing\,,
		\end{equation*}
	and the associated eigenvectors are of the form $ \mathcal{G}(-\lambda^2)\,\mathbf{q} $, with $ \mathbf{q} \in \ker \lf[ \Theta(B) + \Lambda(-\lambda^2) \ri] $.
\end{proposition}

Concerning the scattering, we consider the pair $\big( H_{N}^{(B)}, - \Delta \big)$ and define the wave operators  
	\begin{equation*}
		\Omega_{\pm}\lf(H_{N}^{(B)}, - \Delta \ri) \,:=\ \slim_{t \to \mp \infty}\, e^{i t H_{N}^{(B)}}  e^{i t  \Delta}\,.
	\end{equation*}
	To avoid misunderstandings, let us specify the meaning of completeness for wave operators understood here. Following \cite[\S XI.3]{RS81}, for any pair of self-adjoint operators $A,B$ acting in a given Hibert space $\mathcal{H}$, we say that the wave operators $ \Omega_{\pm}(A,B)  $ are {\it complete} if
			\begin{equation*}
				\mbox{ran}\,\Omega_{\pm}(A,B) = \mbox{ran} \,\Omega_{\pm}(B,A) = \mbox{ran}\,P_{\mathrm{ac}}(A) \,,
			\end{equation*}
		where $ P_{\mathrm{ac}}(A) $ stands for the spectral projector onto the absolute continuity subspace of $ A $. We recall that whenever the wave operators $\Omega_{\pm}(A,B)$ exist, they are complete if and only if $\Omega_{\pm}(B,A)$ exist as well (see \cite[Proposition 3 (vol. III, p. 19)]{RS81}). On the other hand, {\it asymptotic completeness} requires also that $ \sigma_{\mathrm{sc}}(A) = \varnothing $.

\begin{proposition}[Scattering properties]
		\label{cor:ScatteringHbeta}
		\mbox{}	\\
		Let $ N \in \N $ and, for all $n = 1,\dots,N$, let $\alpha_n \in (0,1)$. Then, for any  $B \in M_{2N,\,\mathrm{Herm}}(\mathbb{C}) \cup \lf\{ \infty \ri\}$, the wave operators $\Omega_{\pm}(H_{N}^{(B)}, - \Delta)$ exist and are complete.
\end{proposition}

\section{Proofs}

\subsection{The Friedrichs extension}

We discuss first the properties of the Friedrichs extension, which will play a key role in the analysis of all self-adjoint realizations.

\begin{proof}[Proof of \cref{prop:QF}] 
\mbox{}	

{\it i)} Non-negativity and closedness are obvious consequences of \eqref{eq: QalN} and \eqref{eq: QFr}, respectively. Let us account for \eqref{eq: domQF}, showing the reciprocal inclusion of the sets on its left and right sides. On one hand, we have
	\begin{equation*}
		Q_{N}[\psi] \leqslant 2\, \|\nabla \psi\|_{2}^{2} + 2 \left\|\mbox{$\sum_{n \,=\, 1}^{N}$}\mathbf{A}_{n} \psi \right\|_{2}^{2} \leqslant 2\, \|\nabla \psi\|_{2}^{2} + 2N\,\mbox{$\sum_{n \,=\, 1}^{N}$} \|\mathbf{A}_{n} \psi \|_{2}^{2}\,,
	\end{equation*}
which suffices to infer that the r.h.s. of \eqref{eq: domQF} is a subset of $\dom \big[Q_{N}^{(\mathrm{F})}\big]$. On the other hand, let us consider a partition of unity as in \eqref{eq: xinsupp} and \eqref{eq: sumxin1}: we notice that $\mathbf{S}_n \xi_{n} \in L^{\infty}(\mathbb{R}^2)$ for all $n \in \{0,1,\dots,N\}$. Then, starting again from \eqref{eq: QalN} and using a variant of the IMS localization formula (see, {\it e.g.}, \cite[Thm.\,3.2]{CFKS87}), we derive the following chain of inequalities, for any $ \epsilon \in (0,1)$ and some suitable $ C_\epsilon > 0 $:
	\bmln{
		Q_{N}[\psi] 
			= \tx\sum_{n \,=\, 0}^{N} \lf\| \lf(-i \nabla \!+\! \tx\sum_{m \,=\, 1}^{N} \,\mathbf{A}_{m}\ri)(\xi_{n} \psi) \ri\|_{2}^{2} - \tx\sum_{n \,=\, 0}^{N}  \lf\| \lf(\nabla \xi_{n} \ri)\psi \ri\|_{2}^{2} \\
		\geqslant (1-\epsilon) \, \tx\sum_{n \,=\, 0}^{N} \lf\| \lf(-i \nabla \!+\! \mathbf{A}_{n}\ri)(\xi_{n} \psi) \ri\|_{2}^{2} - \tfrac{1- \epsilon}{\epsilon}\, \tx\sum_{n \,=\, 0}^{N} \lf\| \mathbf{S}_{n} \xi_{n} \psi \ri\|_{2}^{2} - \tx\sum_{n \,=\, 0}^{N}  \lf\| \lf(\nabla \xi_{n} \ri)\psi \ri\|_{2}^{2} \\
		\geqslant (1-\epsilon)^2 \lf\|\nabla (\xi_{0} \psi) \ri\|_{2}^{2} + (1-\epsilon)^2 \,\mbox{$\sum_{n \,=\, 1}^{N}$} (1-\alpha_n)^2 \lf\|\nabla (\xi_{n} \psi) \ri\|_{2}^{2} \\
			+ \epsilon (1-\epsilon) \,\mbox{$\sum_{n \,=\, 1}^{N}$} \min \left\{1 \,, \tfrac{(1 - \alpha_n)^2}{\alpha_n^2}\right\} \lf\| \mathbf{A}_n  \xi_n\psi \ri\|_{2}^{2} - C_{\epsilon} \lf\| \psi \ri\|_{2}^{2} \\
		\geqslant (1-\epsilon)^2 \min_{n\,=\, 1,\dots,N} (1-\alpha_n)^2  \lf\|\nabla \psi \ri\|_{2}^{2} + \epsilon (1-\epsilon) \lf( \min_{n\,=\,1,\dots,N} \left\{1 \,, \tfrac{(1 - \alpha_n)^2}{\alpha_n^2}\right\} \ri) \mbox{$\sum_{n \,=\, 1}^{N}$} \lf\| \mathbf{A}_n \xi_n \psi \ri\|_{2}^{2}  - C_{\epsilon} \lf\| \psi \ri\|_{2}^{2} \,.
	}
Here we have used the lower bounds (see \cite[Eq. (2.4)]{CF21})
	\begin{equation*} 
		\lf\| \lf(-i \nabla + \mathbf{A}_{n}\ri) \psi\ri\|_{2}^{2} \geqslant (1-\alpha_n)^2 \,\lf\|\nabla \psi \ri\|_{2}^{2}\,, \qquad
		\lf\| \lf(-i \nabla + \mathbf{A}_{n}\ri) \psi\ri\|_{2}^{2} \geqslant 
			\min \left\{1 \,, \tfrac{(1 - \alpha_n)^2}{\alpha_n^2}\right\}\,\lf\| \mathbf{A}_n \psi \ri\|_{2}^{2}.
	\end{equation*}
	Summing up, we infer that for any $ c > 0 $ small enough, there exists a finite $ \gamma_c > 0 $ such that
	\begin{equation*}
		Q_{\alpha,N}[\psi] + \gamma_c\,\|\psi\|_{2}^{2} \geqslant c \left( \lf\|\nabla \psi \ri\|_{2}^{2}  + \mbox{$\sum_{n \,=\, 1}^{N}$} \lf\| \mathbf{A}_n\, \xi_n \psi \ri\|_{2}^{2} \right).
	\end{equation*}
This shows that $\dom \big[ Q_{N}^{(\mathrm{F})} \big]$ is a subset of the r.h.s. of \eqref{eq: domQF}, thus proving the identity \eqref{eq: domQF}.

Finally, for any $n \in \{1,\dots,N\}$ the limits in \eqref{eq: limdomQF} can be deduced by the same arguments described in \cite[\S{2.1}]{CF21}, making reference to the representation in polar coordinates centered at $\mathbf{x}_{n}$ and building on the square-integrability of $\nabla \psi,\mathbf{A}_{n} \psi$ near $\mathbf{x}_{n}$ for any $\psi \in \dom \big[ Q_{N}^{(\mathrm{F})} \big]$.

{\it ii)}
Given the sesquilinear form $ Q_{N}^{(\mathrm{F})}[\psi_1,\psi_2] $ defined by polarization starting from the quadratic form $Q_{N}^{(\mathrm{F})}$, the unique operator associated to it (see, {\it e.g.}, \cite[Thm.\,VIII.15]{RS81}) is
	\begin{equation*}
		\dom \big(H_{N}^{(\mathrm{F})}\big) = \lf\{ \psi_{2} \!\in\! \dom\big[ Q_{N}^{(\mathrm{F})} \big] \;\Big|\; \exists\,w \!\in\! L^2(\mathbb{R}^2), Q_{N}^{(\mathrm{F})}[\psi_{1},\psi_{2}] = \braketl{\psi_{1}}{w}_{L^2(\mathbb{R}^2)},\; \forall\, \psi_{1} \!\in\! \dom\big[Q_{N}^{(\mathrm{F})}\big] \ri\}\,,
	\end{equation*}
with $H_{N}^{(\mathrm{F})} \psi_{2} := w$ for all $\psi_{2} \!\in\! \dom \big(H_{N}^{(\mathrm{F})}\big)$. For any $\psi_{1},\psi_{2}\! \in\! \dom\big[Q_{N}^{(\mathrm{F})}\big]$, integrating by parts and noting that $\mathbf{A}_{n} \cdot (\mathbf{x} - \mathbf{x}_{n}) = 0$, we obtain
	\bmln{
		Q_{N}^{(\mathrm{F})}[\psi_{1},\psi_{2}] 
  		= \lim_{r \to 0^+}\! \int_{\mathbb{R}^2 \,\setminus\, \cup_{n = 1}^N B_{r}(\mathbf{x}_n)}\hspace{-0.5cm}\diff\mathbf{x}\; \lf[ \left(-i \nabla + \tx\sum_{\ell \,=\, 1}^N \mathbf{A}_{\ell}\right)\psi_{1} \ri]^* \cdot \left(-i \nabla + \tx\sum_{m \,=\, 1}^N \mathbf{A}_{m}\right)\psi_{2} \\
		= \lim_{r \to 0^+}\! \left[\int_{\mathbb{R}^2 \,\setminus\, \cup_{n = 1}^N B_{r}(\mathbf{x}_n)}\hspace{-0.9cm}\diff\mathbf{x}\; \psi_{1}^{*} \left(-i \nabla \!+\! \tx\sum_{m \,=\, 1}^N \mathbf{A}_{m}\right)^2\! \psi_{2}
- i \sum_{n \,=\, 1}^N \int_{\partial B_{r}(\mathbf{x}_n)}\hspace{-0.6cm}\diff\Sigma_r(\mathbf{x})\; \psi_{1}^{*}\; \tfrac{\mathbf{x} - \mathbf{x}_{n}}{|\mathbf{x} - \mathbf{x}_{n}|} \cdot \left(-i\nabla + \mathbf{S}_{n}\right)\psi_{2} \right] .
	}
Concerning the boundary terms, using the Cauchy-Schwarz inequality and the relations in \eqref{eq: limdomQF}, we infer
	\bmln{
		\left|\int_{\partial B_{r}(\mathbf{x}_n)}\hspace{-0.5cm}\diff\Sigma_r(\mathbf{x})\; \psi_{1}^{*}\; \tfrac{\mathbf{x} - \mathbf{x}_{n}}{|\mathbf{x} - \mathbf{x}_{n}|} \cdot \left(-i\nabla + \mathbf{S}_{n}\right)\psi_{2}\right| \\
		\leqslant
			\lf\| \psi_1 \ri\|_{L^2(\partial B_{r}(\mathbf{x}_n))} \lf( \lf\| \partial_r \psi_2 \ri\|_{L^2(\partial B_{r}(\mathbf{x}_n))} + \lf\|\mathbf{S}_{n}\ri\|_{L^{\infty}\lf(\partial B_r(\mathbf{x}_{n})\ri)} \lf\| \psi_2 \ri\|_{L^2(\partial B_{r}(\mathbf{x}_n)}  \right) \\
		= 2\pi\, r \,\avg{\,|\psi_1|^2}_{n}^{1/2} \lf( \avg{\,|\partial_r \psi_2|^2}_{n}^{1/2} + \lf\|\mathbf{S}_{n}\ri\|_{L^{\infty}(B_r(\mathbf{x}_{n}))} \avg{\,|\psi_2|^2}_{n}^{1/2} \ri)
\;\xrightarrow[r \to 0^+]{}\; 0\,.
	}
Summing up, the above results show that $H_{N}^{(\mathrm{F})} \psi_{1} = H_{N} \psi_{1}$, which in turn yields the thesis.
\end{proof}

We now address the proof of the representation \eqref{eq: RNFried} of the Friedrichs operator resolvent, which is going to play a very important role in the sequel. We henceforth denote by $ \lf\| \: \cdot \: \ri\|_{\mathcal{L}^n} : = \lf\| \: \cdot \: \ri\|_{\mathcal{L}^n(L^2(\mathbb{R}^2))} $, $ n \in [1, + \infty] $, the Schatten norm of order $ n $, and by $ \mathcal{L}^n : = \mathcal{L}^n(L^2(\mathbb{R}^2)) $ the corresponding operators ideal.

	\begin{proposition}[Resolvent of $ H^{(\mathrm{F})}_{N} $]
		\label{lemma: Tn} 
		\mbox{}		\\
		For any $z \in \mathbb{C} \setminus [0,+\infty)$, the operator $T_{N}(z)$ defined in  \eqref{eq: TNdef} is bounded and \eqref{eq: RNFried} holds true, {\it i.e.},
			\begin{equation*}
				R^{(\mathrm{F})}_{N}(z) = \sum_{n = 0}^{N} e^{- i \mathbf{S}_{n}(\mathbf{x}_n) \cdot (\mathbf{x} - \mathbf{x}_n)}\xi_{n}\, R^{(\mathrm{F})}_{n}(z)\, \xi_{n}\,e^{i \mathbf{S}_{n}(\mathbf{x}_n) \cdot (\mathbf{x} - \mathbf{x}_n)}\, \big[1 + T_N(z)\big]^{-1} \,. 
			\end{equation*}
	\end{proposition}

We first present an auxiliary lemma. Recall the definition \eqref{eq: TNdef} of the operator $ T_N(z) $:
	\begin{equation*}
		T_N(z)  
		= \tx\sum_{n} e^{-i \mathbf{S}_{n}(\mathbf{x}_n) \cdot (\mathbf{x} - \mathbf{x}_n)} P_n\, R^{(\mathrm{F})}_{n}(z)\, \xi_{n}\,e^{i \mathbf{S}_{n}(\mathbf{x}_n) \cdot (\mathbf{x} - \mathbf{x}_n)}\,,
	\end{equation*}
where $ P_n = 2 \,\lf(\check{\mathbf{S}}_{n} \xi_{n} - i \nabla \xi_{n}\ri) \!\cdot\! (-i\nabla \!+\! \mathbf{A}_{n})
				+ \check{\mathbf{S}}_{n}^2 \xi_{n} + 2 \,\check{\mathbf{S}}_{n}\! \cdot ( - i \nabla \xi_{n})\! - \Delta \xi_{n} $.
				
	\begin{lemma}
		\label{lemma: TN small}
		\mbox{}	\\
		The operator $ T_N(z)$ is bounded for any $ z \in \mathbb{C} \setminus \R^+ $ and $ \lf\|T_N(z) \ri\| < 1$ for any $ z \in \mathbb{C} $ such that $\dist(z, \mathbb{R}^{+}) $ is large enough.
	\end{lemma}
	
	\begin{proof}
		Since $\sigma\big(H_{n}^{(\mathrm{F})}\big) = [0,+\infty)$, we have the bound 
		\begin{equation*}
			\lf\|R^{(\mathrm{F})}_{n}(z) \ri\| \leqslant \frac{1}{ \dist(z,\mathbb{R}^{+})}\,,
		\end{equation*}
		which together with the fact that $\mbox{ran} \big(R^{(\mathrm{F})}_{n}(z)\big) \subset \dom \big[Q_{N}^{(\mathrm{F})}\big]$ and the asymptotic conditions \eqref{eq: limdomQF}, yields
	\bmln{
	\big\|(-i\nabla \!+\! \mathbf{A}_{n}) R^{(\mathrm{F})}_{n}(z)\, \psi \big\|_{2}^2 
		= \lim_{r \to 0^{+}} \int_{\mathbb{R}^2 \setminus B_{r}(\mathbf{x}_n)}\hspace{-0.4cm} \diff\mathbf{x} \left|(-i\nabla \!+\! \mathbf{A}_{n}) R^{(\mathrm{F})}_{n}(z)\psi \right|^2 \\	
	= \lim_{r \to 0^{+}} \left[ - \int_{\partial B_{r}(\mathbf{x}_n)}\hspace{-0.6cm} \diff\Sigma_r\,  \lf( R^{(\mathrm{F})}_{n}(z)\psi \ri)^*\;\partial_r \big(R^{(\mathrm{F})}_{n}(z)\psi\big) + \int_{\mathbb{R}^2 \setminus B_{r}(\mathbf{x}_n)}\hspace{-0.6cm} \diff\mathbf{x}\; \lf( R^{(\mathrm{F})}_{n}(z)\psi \ri)^* \; (-i\nabla \!+\! \mathbf{A}_{n})^2 R^{(\mathrm{F})}_{n}(z)\psi
	\right] \\
		= \big\langle R^{(\mathrm{F})}_{n}(z)\,\psi\,\big|\, \psi \big\rangle + z\,\big\|R^{(\mathrm{F})}_{n}(z)\,\psi\,\big\|_{2}^2 
		\leqslant  \tfrac{1}{\dist (z,\mathbb{R}^{+})}\left(1 + \tfrac{|z|}{\dist (z,\mathbb{R}^{+})} \right) \|\psi\|_2^{2}\,.	
	}
Since $\xi_{n}$ is smooth and uniformly bounded together with all its derivatives, and the same can be said for $\check{\mathbf{S}}_n$ on $\supp \xi_{n}$, from the above arguments we readily infer\footnote{Here and in the following we denote by $ C $ a positive finite constant, whose value may change from line to line.}
	\begin{equation*}
		\lf\|T_n(z) \ri\| \leqslant \frac{C}{\dist(z,\mathbb{R}^{+})} \left(1 + \frac{|z|}{\dist(z,\mathbb{R}^{+})} \right) ,
	\end{equation*}
for all $z \in \mathbb{C} \setminus [0,+\infty)$, which proves the thesis.
	\end{proof}

	\begin{proof}[Proof of \cref{lemma: Tn}]
		We start from the ansatz
		\begin{equation}\label{eq: RNansatz}
			R^{(\mathrm{F})}_{N}(z) = \sum_{n \,=\, 0}^{N} \, e^{- i \mathbf{S}_{n}(\mathbf{x}_n) \cdot (\mathbf{x} - \mathbf{x}_n)}\xi_{n}\, R^{(\mathrm{F})}_{n}(z)\, \xi_{n}\, e^{i \mathbf{S}_{n}(\mathbf{x}_n) \cdot (\mathbf{x} - \mathbf{x}_n)} + W_N(z)\,,
		\end{equation}
where $W_{N}(z)$ is a suitable reminder operator to be determined {\it a posteriori}. 
Applying $H_{N}^{(\mathrm{F})} - z$ on both sides of \eqref{eq: RNansatz}, we get
	\bmln{
		1 =  \tx\sum_{n} \,  e^{- i \mathbf{S}_{n}(\mathbf{x}_n) \cdot (\mathbf{x} - \mathbf{x}_n)} \Big[
			\xi_{n} \Big(\! \big(\! - i \nabla \!+\! \mathbf{A}_{n} \!+\! \check{\mathbf{S}}_{n} \big)^2\! - z \Big) R^{(\mathrm{F})}_{n}(z)\, \xi_{n} \\
			+ 2 ( - i \nabla \xi_{n}) \!\cdot\! \big(\!-i\nabla \!+\! \mathbf{A}_{n} \!+\! \check{\mathbf{S}}_{n} \big) R^{(\mathrm{F})}_{n}(z)\, \xi_{n}
			- (\Delta \xi_{n})\, R^{(\mathrm{F})}_{n}(z)\, \xi_{n}  \Big]\, e^{i \mathbf{S}_{n}(\mathbf{x}_n) \cdot (\mathbf{x} - \mathbf{x}_n)}
			+ \big( H_{N}^{(\mathrm{F})} - z \big) W_N(z) \\
		= \tx\sum_{n} \,  e^{- i \mathbf{S}_{n}(\mathbf{x}_n) \cdot (\mathbf{x} - \mathbf{x}_n)} \Big[ \xi_{n} \big( ( - i \nabla \!+\! \mathbf{A}_{n})^2 - z\big)
			+ 2 \,(\check{\mathbf{S}}_{n} \xi_{n} - i \nabla \xi_{n}) \!\cdot\! (-i\nabla \!+\! \mathbf{A}_{n}) \\
			+ \check{\mathbf{S}}_{n}^2 \xi_{n} + 2 \,\check{\mathbf{S}}_{n}\! \cdot ( - i \nabla \xi_{n})\! - \Delta \xi_{n}
			\Big] R^{(\mathrm{F})}_{n}(z)\, \xi_{n}\,e^{i \mathbf{S}_{n}(\mathbf{x}_n) \cdot (\mathbf{x} - \mathbf{x}_n)}
			+ \big( H_{N}^{(\mathrm{F})}\! - z \big) W_N(z)\,.
	}
Since $ (( - i \nabla \!+\! \mathbf{A}_{n})^2 - z ) R^{(\mathrm{F})}_{n}(z) = 1 $ and $\sum_{n = 0}^{N} \xi_{n}^2 = 1$, we infer that
	\bmln{
		W_N(z) = -\,R^{(\mathrm{F})}_{N}(z) \; \tx\sum_{n} \, e^{-i \mathbf{S}_{n}(\mathbf{x}_n) \cdot (\mathbf{x} - \mathbf{x}_n)} \Big[
					2 \,(\check{\mathbf{S}}_{n} \xi_{n} - i \nabla \xi_{n}) \!\cdot\! (-i\nabla \!+\! \mathbf{A}_{n}) \\
					+ \check{\mathbf{S}}_{n}^2 \xi_{n} + 2 \,\check{\mathbf{S}}_{n}\! \cdot ( - i \nabla \xi_{n})\! - \Delta \xi_{n}
				\Big] R^{(\mathrm{F})}_{n}(z)\, \xi_{n}\, e^{i \mathbf{S}_{n}(\mathbf{x}_n) \cdot (\mathbf{x} - \mathbf{x}_n)}\,.
	}
Summing up, we obtain
	\begin{equation*}
		R^{(\mathrm{F})}_{N}(z) \big[ 1 + T_N(z) \big] = \tx\sum_{n} \, e^{-i \mathbf{S}_{n}(\mathbf{x}_n) \cdot (\mathbf{x} - \mathbf{x}_n)} \xi_{n}\, R^{(\mathrm{F})}_{n}(z)\, \xi_{n}\,e^{i \mathbf{S}_{n}(\mathbf{x}_n) \cdot (\mathbf{x} - \mathbf{x}_n)} \,, 
	\end{equation*}
and it just remains to point out that, due to \cref{lemma: TN small}, $ 1  + T_N(z) $ is invertible via the Neumann series
		\begin{equation*}
		\big[1 + T_N(z)\big]^{-1} = \sum_{j = 0}^{+\infty}\, (-1)^{j} \big(T_N(z)\big)^{j} 	\,.
		\end{equation*}
	\end{proof}

We now address the spectral and scattering properties of the Friedrichs realization, which will serve as basic ingredients of the proofs of \cref{pro: spectrum} and \cref{cor:ScatteringHbeta}.

\begin{proposition}[Spectral properties of  $ H_N^{(\mathrm{F})} $]
	\label{pro: spectrum F}
	\mbox{}	\\
	Let $ N \in \N $ and, for all $n = 1,\dots,N$, let $\alpha_n \in (0,1)$. Then, 
	\begin{equation*}
		\sigma_{\mathrm{ess}}\big(H_{N}^{(\mathrm{F})}\big) = [0,+\infty)\,,
		\qquad
		\sigma_{\mathrm{disc}}\big(H_{N}^{(\mathrm{F})}\big) = \varnothing\,.
	\end{equation*}
\end{proposition}

\begin{proposition}[Scattering properties of  $ H_N^{(\mathrm{F})} $]
		\label{thm: OmHNH0}
		\mbox{}		\\
		Let $ N \in \N $ and, for all $n = 1, \ldots, N$, let $\alpha_n \in (0,1)$. Then, the wave operators $\Omega_{\pm}\big(H_{N}^{(\mathrm{F})}, - \Delta \big)$ exist and are complete and, in addition, 
		\begin{equation*}
		\sigma_{\mathrm{ac}}\big(H_{N}^{(\mathrm{F})}\big) = \sigma_{\mathrm{ac}}(-\Delta) = [0,+\infty)\,.
		\end{equation*}
\end{proposition}

We state first some technical lemmas which will be used in the proofs of the main results. Our arguments involve an auxiliary magnetic Schr\"odinger operator, namely,
	\begin{equation}\label{eq:HAreg}
		H_{\mathbf{A}} = (-i \nabla + \mathbf{A})^2\,,
	\end{equation}
where $\mathbf{A} \in C^\infty(\mathbb{R}^2,\mathbb{R}^2)$ is any given vector potential fulfilling
	\begin{equation}\label{eq:AAasy}
		\mathbf{A} = \sum_{n\,=\,1}^{N} \mathbf{A}_{n} \qquad \mbox{in\; $\mathbb{R}^2 \!\setminus\! \mathcal{Q}$}\,,
	\end{equation}
for some bounded open subset $\mathcal{Q} \subset \mathbb{R}^2$ such that $\mathbf{x}_{n} \in \mathcal{Q}$ for all $n \in \{ 1,\dots ,N\}$. 
In particular, under these hypotheses, we certainly have
	\begin{equation*}
		\mathbf{A} \in L^2_{\mathrm{loc}}(\mathbb{R}^2) \cap L^{\infty}(\mathbb{R}^2)\,,
	\end{equation*}
and
	\begin{equation*}
		\mathbf{A}(\mathbf{x}) = \left( \,\sum_{n\,=\,1}^{N} \alpha_n \right) {\mathbf{x}^{\perp} \over |\mathbf{x}|^2} + \mathcal{O}\!\lf({1 \over |\mathbf{x}|^2} \ri) , \qquad \mbox{for\; $|\mathbf{x}| \to +\infty$}\,.
	\end{equation*}
Without loss of generality, we further assume $\mathbf{A}$ to fulfill the Coulomb gauge
	\begin{equation}\label{eq:AACou}
		\nabla \cdot \mathbf{A} = 0 \quad \mbox{in\, $\mathbb{R}^2$}.
	\end{equation}

	\begin{lemma}
		\label{lemma: HA}
		\mbox{}		\\
		For any $\mathbf{A} \in C^\infty(\mathbb{R}^2,\mathbb{R}^2)$ fulfilling \eqref{eq:AAasy} and \eqref{eq:AACou}, the operator $H_{\mathbf{A}}$ defined in \eqref{eq:HAreg} is essentially self-adjoint on $C^{\infty}_{\mathrm{c}}(\mathbb{R}^2)$. Its unique self-adjoint realization is the positive operator
		\begin{equation}\label{eq: HAdom}
			H_{\mathbf{A}}^{(\mathrm{F})} := (-i \nabla + \mathbf{A})^2\,, \qquad 
			\dom\big(H_{\mathbf{A}}^{(\mathrm{F})}\big) = H^2(\mathbb{R}^2)\,.
		\end{equation}
		For any $ z \in \mathbb{C} $ such that $\dist(z, \mathbb{R}^{+}) $ is large enough, the resolvent operator $R_{\mathbf{A}}^{(\mathrm{F})}(z) := \big(H_{\mathbf{A}}^{(\mathrm{F})} - z \big)^{-1}$ satisfies
		\begin{equation}
			R_{\mathbf{A}}^{(\mathrm{F})}(z) 
			= R_{0}(z) \lf[1 + \lf(H_{\mathbf{A}}^{(\mathrm{F})} + \Delta \ri) R_{0}(z) \ri]^{-1}
			= \lf[1 + R_{0}(z) \lf(H_{\mathbf{A}}^{(\mathrm{F})} + \Delta \ri) \ri]^{-1}\! R_{0}(z) \,. \label{eq: RAR0}
		\end{equation}
	\end{lemma}
	
\begin{proof}
	The first part of the thesis follows noting that $H_{\mathbf{A}}$ is a small perturbation of the free Laplacian in the sense of Kato. In fact, for any $\psi \in  H^2(\mathbb{R}^2)$ and any $\varepsilon \in (0,1)$, we have
	\begin{equation*}
		\lf\| \lf( H_{\mathbf{A}} + \Delta \ri) \psi \ri\|_{2} 
		\leqslant 2\, \lf\| \mathbf{A} \ri\|_{\infty} \lf\| \nabla \psi \ri\|_{2} + \lf\| \mathbf{A} \ri\|_{\infty}^{2}\, \lf\|\psi \ri\|_{2}
		\leqslant \varepsilon\, \big\| (-\Delta) \psi \big\|_{2} + \lf(\tfrac{1}{\varepsilon} + 1 \ri) \lf\| \mathbf{A} \ri\|_{\infty}^2\, \lf\| \psi \ri\|_{2}\,.
	\end{equation*}
	Furthermore, $H_{\mathbf{A}}^{(\mathrm{F})}$ is positive definite, so that $\sigma\big(H_{\mathbf{A}}^{(\mathrm{F})}\big) \subseteq [0,+\infty)$ and $R_{\mathbf{A}}^{(\mathrm{F})}(z) $ is a bounded operator for any $z \in \mathbb{C} \setminus [0,+\infty)$. By means of the second resolvent identity, we infer
	\begin{equation}\label{eq: 2RId}
		R_{\mathbf{A}}^{(\mathrm{F})}(z) \lf[ 1 + \lf(H_{\mathbf{A}}^{(\mathrm{F})}+ \Delta \ri) R_{0} (z)\ri] = R_{0}(z)  \,. 
	\end{equation}
Recalling that $\mathbf{A} \in L^{\infty}(\mathbb{R}^2)$, by elementary arguments we get
	\bmln{
		\lf\| \lf(H_{\mathbf{A}}^{(\mathrm{F})}  + \Delta \ri) R_{0}(z) \ri\| 
		= \lf\| \lf(2 \mathbf{A} \cdot (-i \nabla) + \lf| \mathbf{A} \ri|^2\ri) R_{0}(z) \ri\| \\
		\leqslant 2\, \|\mathbf{A}\|_{\infty} \,\esssup_{\mathbf{k} \in \mathbb{R}^2} \lf| \tfrac{\mathbf{k}}{|\mathbf{k}|^2 - z} \ri| + \|\mathbf{A}\|_{\infty}^2\, \esssup_{\mathbf{k} \in \mathbb{R}^2} \lf| \tfrac{1}{|\mathbf{k}|^2 - z} \ri| 
		\leqslant C \lf\|\mathbf{A} \ri\|_{\infty}  \left( \tfrac{1}{\sqrt{|z| - \Re z}} + \tfrac{
\lf\| \mathbf{A} \ri\|_{\infty}}{\dist\lf(z, \mathbb{R}^{+}\ri)} \ri) .
	}
From here we deduce $\big\| \big(H_{\mathbf{A}}^{(\mathrm{F})} + \Delta \big) R_{0}(z) \big\| < 1$, for any $z$ far enough from $ \mathbb{R}^+$, which in turn ensures that $1 + \big(H_{\mathbf{A}}^{(\mathrm{F})} + \Delta\big)R_{0}^{(\mathrm{F})}(z)$ is indeed invertible. In view of this, the first identity in \eqref{eq: RAR0} follows readily from \eqref{eq: 2RId}. The second identity in \eqref{eq: RAR0} can be derived by evaluating the adjoint of the first identity, with $z^*$ in place of  $ z $.
\end{proof}

	\begin{lemma}
		\label{lemma: RAR0compact}
		\mbox{}		\\
		For any $\mathbf{A} \in C^\infty(\mathbb{R}^2,\mathbb{R}^2)$ fulfilling \eqref{eq:AAasy} and \eqref{eq:AACou}, there holds
		\begin{equation*}
			R_{\mathbf{A}}^{(\mathrm{F})}(z) - R_{0}(z) \in \mathcal{L}^{\infty}\,, \qquad  \forall\, z \in \mathbb{C} \setminus [0,+\infty) \,,
		\end{equation*}
		so that 
		\begin{equation*}
				\sigma_{\mathrm{ess}}\big(H_{\mathbf{A}}^{(\mathrm{F})}\big) = \sigma_{\mathrm{ess}}( - \Delta) = [0,+\infty)\,,	\qquad		\sigma_{\mathrm{disc}}\big(H_{\mathbf{A}}^{(\mathrm{F})}\big) = \varnothing\,.
		\end{equation*}
	\end{lemma}
	
\begin{proof}
It suffices to prove the thesis for some $z$ in the resolvent set. Let us fix $z = - \lambda^2$, for some $\lambda > 0$, and proceed to notice that the second resolvent identity yields
	\begin{equation}\label{eq: comp0}
		R_{\mathbf{A}}^{(\mathrm{F})}(-\lambda^2) - R_{0}(-\lambda^2) = 
		- 2\,R_{0}(-\lambda^2)\, \mathbf{A} \cdot (-i \nabla) R_{\mathbf{A}}^{(\mathrm{F})}(-\lambda^2)
		- R_{0}(-\lambda^2) \lf| \mathbf{A} \ri|^2 R_{\mathbf{A}}^{(\mathrm{F})}(-\lambda^2)\,.
	\end{equation}
Hereafter, we show that both addenda on the right-hand side of \eqref{eq: comp0} are indeed compact operators.

We introduce the one-parameter family of smeared vector potentials
	\begin{equation*}
		\mathbf{A}_{\varepsilon}(\mathbf{x}) := {1 \over (1 + |\mathbf{x}|)^{\varepsilon}}\,\mathbf{A}(\mathbf{x}) \qquad (\eps > 0)\,.
	\end{equation*}
	Since $|\mathbf{A}(\mathbf{x})| \leqslant {C \over 1+|\mathbf{x}|}$, we have
	\begin{equation*}
		\left\|R_{0}(-\lambda^2)\, \mathbf{A}_{\varepsilon}\right\|_{\mathcal{L}^2}^2 
				\leqslant C \lf\| R_{0}(-\lambda^2) \ri\|_2^2  \int_{\mathbb{R}^2}\hspace{-0.1cm} \diff \mathbf{\xv}\; {1 \over (1 + |\mathbf{x}|)^{2+2\varepsilon}}  < + \infty\,,
	\end{equation*}
proving that $R_{0}(-\lambda^2)\, \mathbf{A}_{\varepsilon}$ is a Hilbert-Schmidt operator. To say more, we have
	\begin{equation*}
		\lf\|\mathbf{A}_{\varepsilon} - \mathbf{A}\ri\|_{\infty} 
		\leqslant C \,\esssup_{\mathbf{x} \in \mathbb{R}^2} \tfrac{(1 + |\mathbf{x}|)^{\varepsilon} - 1}{(1 + |\mathbf{x}|)^{\varepsilon} (1 + |\mathbf{x}|)}
		= C \,\varepsilon\,(1+ \varepsilon)^{-{1 + \varepsilon \over \varepsilon}} \,\xrightarrow[\varepsilon \to 0^{+}]{}\, 0\,,
	\end{equation*}
which implies, in turn, that $R_{0}(-\lambda^2)\, \mathbf{A}_{\varepsilon}$ converges to $R_{0}(-\lambda^2)\, \mathbf{A}$ in norm. Since  $(-i \nabla) R_{\mathbf{A}}^{(\mathrm{F})}(-\lambda^2)$ is bounded, we deduce that the first operator on the r.h.s. of \eqref{eq: comp0} is compact.
On the other hand,  
	\begin{equation*}
		\left\|R_{0}(- \lambda^2)\, \lf| \mathbf{A} \ri|^2\right\|_{\mathcal{L}^2}^{2} 
		\leqslant C \lf\| R_{0}(-\lambda^2) \ri\|_2^2  \int_{\mathbb{R}^2}\hspace{-0.1cm} \diff\mathbf{\xv}\; {1 \over (1+|\mathbf{x}|)^4} < + \infty\,.
	\end{equation*}
Hence, since $R_{\mathbf{A}}^{(\mathrm{F})}(-\lambda^2)$ is bounded, the second operator on the r.h.s. of \eqref{eq: comp0} is compact too.

A straightforward application of Weyl's criterion \cite[Thm. XIII.14]{RS81} completes the proof.
\end{proof}

Next, we recall an important result about scattering theory for magnetic Schr\"odinger operators originally proved in \cite{LT87} and \cite{Ta99}.
	
	\begin{proposition}\label{prop: RAR0scatt}
		\mbox{}	\\
		For any $\mathbf{A} \in C^\infty(\mathbb{R}^2,\mathbb{R}^2)$ fulfilling \eqref{eq:AAasy} and \eqref{eq:AACou}, the wave operators $ \Omega_{\pm} (H_{\mathbf{A}}^{(\mathrm{F})}, - \Delta  ) $ exist and are asymptotically complete, so that
		\begin{equation*}
				\sigma_{\mathrm{ac}}\big(H_{\mathbf{A}}^{(\mathrm{F})}\big) =  \sigma_{\mathrm{ac}}\lf(-\Delta\ri) = [0,+\infty)\,,
				\qquad 
				\sigma_{\mathrm{sc}}\big( H_{\mathbf{A}}^{(\mathrm{F})}\big) = \varnothing\,.
		\end{equation*}
	\end{proposition}
	\begin{proof}
		See, {\it e.g.}, \cite[Theorem 2]{LT87}.
	\end{proof}

	\begin{lemma}
		\label{lemma: RNRAS2}
		\mbox{}		\\
		Let $ N \in \N $ and, for all $n = 1, \ldots, N$, let $\alpha_n \in (0,1)$. Then, for any $\mathbf{A} \in C^\infty(\mathbb{R}^2,\mathbb{R}^2)$ fulfilling \eqref{eq:AAasy} and \eqref{eq:AACou}, there holds
			\begin{equation*}
				R_{N}^{(\mathrm{F})}(z) - R_{\mathbf{A}}^{(\mathrm{F})}(z) \in \mathcal{L}^{2} , \qquad  \forall\, z \in \mathbb{C} \setminus [0,+\infty) \,.
			\end{equation*}
	\end{lemma}

\begin{proof}
Let $\one_{\mathcal{Q}}$ be the indicator function associated to the bounded set $\mathcal{Q}$ appearing in \eqref{eq:AAasy}. An elementary computation shows that
	\begin{equation}\label{eq: RNRAunoO}
		H_{N}^{(\mathrm{F})} - H_{\mathbf{A}}^{(\mathrm{F})} 
		= \big(\mbox{$\sum_n$}\mathbf{A}_n - \mathbf{A} \big) \cdot \big[ 2 (- i\nabla) + \mbox{$\sum_n$}\mathbf{A}_n + \mathbf{A} \big]
		= \one_{\mathcal{Q}} \lf(H_{N}^{(\mathrm{F})} - H_{\mathbf{A}}^{(\mathrm{F})} \ri).
	\end{equation}
Then, using the second resolvent identity and the representation \eqref{eq: RAR0} for $R_{\mathbf{A}}^{(\mathrm{F})}(z)$, we deduce
	\begin{equation*}
		R_{N}^{(\mathrm{F})}(z) - R_{\mathbf{A}}^{(\mathrm{F})}(z) 
		= - \lf[1 + R_{0} (z) \lf(H_{\mathbf{A}}^{(\mathrm{F})} + \Delta \ri) \ri]^{-1} R_{0} (z)\, \one_{\mathcal{Q}}\, \lf( H_{N}^{(\mathrm{F})} - H_{\mathbf{A}}^{(\mathrm{F})} \ri)  R_{N}^{(\mathrm{F})}(z)\,.
	\end{equation*}
Notice that $\lf[\one + R_{0} (z) \lf(H_{\mathbf{A}}^{(\mathrm{F})} + \Delta \ri) \ri]^{-1}$ and $\big[ H_{N}^{(\mathrm{F})} - H_{\mathbf{A}}^{(\mathrm{F})} \big] R_{N}^{(\mathrm{F})}(z)$ are bounded operators. On the other hand, 
	\begin{equation*}
		\lf\|R_{0}(-\lambda^2) \one_{\mathcal{Q}}\ri\|_{\mathcal{L}^2}^{2} \leq \lf\| R_{0}(-\lambda^2) \ri\|_2^2 \lf| \Omega \ri| < + \infty\,.
	\end{equation*}
This suffices to infer that $R_{0}^{(\mathrm{F})}(z)\, \one_{\mathcal{Q}} \in \mathcal{L}^2$ for $z = -\lambda^2$, whence for any $z \in \mathbb{C} \setminus [0,+\infty)$, which in turn implies the thesis.
\end{proof}

	\begin{lemma}
		\label{lemma: RN2RA2S1}
		\mbox{}		\\
		Let $ N \in \N $ and, for all $n = 1, \ldots, N$, let $\alpha_n \in (0,1)$. Then, for any $\mathbf{A} \in C^\infty(\mathbb{R}^2,\mathbb{R}^2)$ fulfilling \eqref{eq:AAasy} and \eqref{eq:AACou}, there holds
		\begin{equation*}
			\big(R_{N}^{(\mathrm{F})}(z)\big)^2 - \big(R_{\mathbf{A}}^{(\mathrm{F})}(z)\big)^2 \in \mathcal{L}^1 , \qquad  \forall\, z \in \mathbb{C} \setminus [0,+\infty)\,.
		\end{equation*}
	\end{lemma}
	
\begin{proof}
To begin with, let us draw the attention to the elementary algebraic identity
	\bmln{
		\big(R_{N}^{(\mathrm{F})}(z)\big)^2 - \big(R_{\mathbf{A}}^{(\mathrm{F})}(z)\big)^2 
		= \lf(R_{N}^{(\mathrm{F})}(z) - R_{\mathbf{A}}^{(\mathrm{F})}(z)\ri)^2
			+ \lf(R_{N}^{(\mathrm{F})}(z) - R_{\mathbf{A}}^{(\mathrm{F})}(z)\ri) R_{\mathbf{A}}^{(\mathrm{F})}(z)
			\\
			+ R_{\mathbf{A}}^{(\mathrm{F})}(z) \lf(R_{N}^{(\mathrm{F})}(z) - R_{\mathbf{A}}^{(\mathrm{F})}(z)\ri) .
	}
Since $R_{N}^{(\mathrm{F})}(z) - R_{\mathbf{A}}^{(\mathrm{F})}(z) \in \mathcal{L}^2$ by \cref{lemma: RNRAS2}, using basic properties of Hilbert-Schmidt operators we readily obtain that $\big(R_{N}^{(\mathrm{F})}(z) - R_{\mathbf{A}}^{(\mathrm{F})}(z)\big)^2 \in \mathcal{L}^1 $ (see, {\it e.g.}, \cite[Thm. VI.22]{RS81}). For $ z $ real the other two terms are one the adjoint of the other, so it suffices to prove that
	\begin{equation*}
		R_{\mathbf{A}}^{(\mathrm{F})}(z) \lf(R_{N}^{(\mathrm{F})}(z) - R_{\mathbf{A}}^{(\mathrm{F})}(z)\ri) \in \mathcal{L}^1 .
	\end{equation*}
	Let $\chi_{\mathcal{Q}} \in C^{\infty}_{\mathrm{c}}(\mathbb{R}^2)$ be such that $\chi_{\mathcal{Q}} \equiv 1$ in $\mathcal{Q}$. Acting as in the derivation of \eqref{eq: RNRAunoO}, we get
	\begin{equation*}
		R_{\mathbf{A}}^{(\mathrm{F})}(z) \lf(R_{N}^{(\mathrm{F})}(z) - R_{\mathbf{A}}^{(\mathrm{F})}(z)\ri) 
		= -\, \big(R_{\mathbf{A}}^{(\mathrm{F})}(z)\big)^2\, \chi_{\mathcal{Q}} \lf(H_{N}^{(\mathrm{F})} - H_{\mathbf{A}}^{(\mathrm{F})}\ri) R_{N}^{(\mathrm{F})}(z)\,.
	\end{equation*}
Since $\big(H_{N}^{(\mathrm{F})} - H_{\mathbf{A}}^{(\mathrm{F})}\big) R_{N}^{(\mathrm{F})}(z)$ is a bounded operator, it remains to show that
	\begin{equation}\label{eq: RAchiS1}
		\big(R_{\mathbf{A}}^{(\mathrm{F})}(z)\big)^2\, \chi_{\mathcal{Q}} \in \mathcal{L}^1 \,.
	\end{equation}
To this purpose, we repeatedly exploit the following identity, which is a straightforward consequence of the Coulomb gauge,
\bmln{
		\lf[(-i\partial_{j}), H_{\mathbf{A}}^{(\mathrm{F})} \ri]
		= \lf[(-i\partial_{j}), \lf|\mathbf{A}\ri|^2 \ri] + 2 \sum_{k = 1}^2 \lf[(-i\partial_{j}), \mathbf{A}_k \ri] (-i\partial_k)   
		=  - i \partial_j \lf| \mathbf{A} \ri|^2 + 2 \sum_{k = 1}^2 (-i \partial_j \mathbf{A}_k) (-i\partial_k)  \\
		= - i \partial_j \lf| \mathbf{A} \ri|^2   + 2 \lf(- i \nabla \ri) \cdot (-i \partial_j \mathbf{A})\,,
		}
to compute	
	\bml{\label{eq: seqid}
		\big(R_{\mathbf{A}}^{(\mathrm{F})}(z)\big)^2 \chi_{\mathcal{Q}} 
		= - \,\big(R_{\mathbf{A}}^{(\mathrm{F})}(z)\big)^2\big[H_{\mathbf{A}}^{(\mathrm{F})}, \chi_{\mathcal{Q}}\big] R_{\mathbf{A}}^{(\mathrm{F})}(z) + R_{\mathbf{A}}^{(\mathrm{F})}(z)\,\chi_{\mathcal{Q}} \, R_{\mathbf{A}}^{(\mathrm{F})}(z)  \\
		=  -  \,2 \big(R_{\mathbf{A}}^{(\mathrm{F})}(z)\big)^2 (-i \nabla) \cdot (- i \nabla \chi_{\mathcal{Q}}) R_{\mathbf{A}}^{(\mathrm{F})}(z) 
			-\, \big(R_{\mathbf{A}}^{(\mathrm{F})}(z)\big)^2 \big[2 \mathbf{A} \cdot (- i \nabla \chi_{\mathcal{Q}}) + \Delta \chi_{\mathcal{Q}} \big] R_{\mathbf{A}}^{(\mathrm{F})}(z)	\\
			+ R_{\mathbf{A}}^{(\mathrm{F})}(z)\,\chi_{\mathcal{Q}} \, R_{\mathbf{A}}^{(\mathrm{F})}(z)	  \\
		= - \,4 \sum_{j =1}^2 \lf\{ \big(R_{\mathbf{A}}^{(\mathrm{F})}(z) \big)^2 \big[ 2  (-i \partial_j \mathbf{A}) \cdot (-i \nabla)  	
		+ \big(- i \partial_j \lf|\mathbf{A}\ri|^2\big) \big] R_{\mathbf{A}}^{(\mathrm{F})}(z) (-i \nabla \mathbf{A}_j) \cdot R_{\mathbf{A}}^{(\mathrm{F})}(z) (- i \nabla \chi_{\mathcal{Q}}) R_{\mathbf{A}}^{(\mathrm{F})}(z)	\ri. \\
		+ \,2\, R_{\mathbf{A}}^{(\mathrm{F})}(z) \big[ 2 (-i\partial_j) (-i \nabla \mathbf{A}_j)
		+ \big(- i \nabla \lf|\mathbf{A}\ri|^2\big)\big] \cdot \big(R_{\mathbf{A}}^{(\mathrm{F})}(z) \big)^2 (- i \nabla \chi_{\mathcal{Q}}) R_{\mathbf{A}}^{(\mathrm{F})}(z) \\
		\lf. -\,4\, R_{\mathbf{A}}^{(\mathrm{F})}(z) (-i\partial_j) R_{\mathbf{A}}^{(\mathrm{F})}(z)\, (-i \nabla \mathbf{A}_j) \cdot R_{\mathbf{A}}^{(\mathrm{F})}(z) (- i \nabla \chi_{\mathcal{Q}}) R_{\mathbf{A}}^{(\mathrm{F})}(z) \ri\}	\\		
			-\,2 \big(R_{\mathbf{A}}^{(\mathrm{F})}(z)\big)^2 \big(- i \nabla \lf|\mathbf{A}\ri|^2 \big) \cdot\, R_{\mathbf{A}}^{(\mathrm{F})}(z) (- i \nabla \chi_{\mathcal{Q}}) R_{\mathbf{A}}^{(\mathrm{F})}(z)   - \,2\, (-i \nabla) \cdot \big(R_{\mathbf{A}}^{(\mathrm{F})}(z)\big)^2 (- i \nabla \chi_{\mathcal{Q}}) R_{\mathbf{A}}^{(\mathrm{F})}(z) \\
			-\, \big(R_{\mathbf{A}}^{(\mathrm{F})}(z)\big)^2 \big[2 \mathbf{A} \cdot (- i \nabla \chi_{\mathcal{Q}}) + \Delta \chi_{\mathcal{Q}}  \big] R_{\mathbf{A}}^{(\mathrm{F})}(z)	
			+ R_{\mathbf{A}}^{(\mathrm{F})}(z)\,\chi_{\mathcal{Q}} \, R_{\mathbf{A}}^{(\mathrm{F})}(z)\,.
	}
Given that $\mathbf{A}$, $\chi$ and all their derivatives are uniformly bounded in $\mathbb{R}^2$, they all define bounded operators. Furthermore, since $\dom\big(H_{\mathbf{A}}^{(\mathrm{F})}\big) = H^2(\mathbb{R}^2)$ (see \eqref{eq: HAdom}), we certainly have $(-i \nabla) R_{\mathbf{A}}^{(\mathrm{F})}(z) \in \mathcal{B}(L^2(\mathbb{R}^2)) $. Taking this into account, it can be checked by direct inspection that each addendum in the last expression of \eqref{eq: seqid} is indeed a product of bounded operators including at least one term of the form
	\begin{equation*}
		R_{\mathbf{A}}^{(\mathrm{F})}(z)\,\mathcal{X}\, R_{\mathbf{A}}^{(\mathrm{F})}(z) \,, \qquad \mbox{for some $\mathcal{X} \in C^{\infty}_{\mathrm{c}}(\mathbb{R}^2)$}\,.
	\end{equation*}
We claim that any such term belongs to $ \mathcal{L}^1 $. As a matter of fact, in view of \eqref{eq: RAR0}, it suffices to prove that
	\begin{equation}\label{eq: R0XR0}
		R_{0}(-\lambda^2)\,\mathcal{X}\, R_{0}(-\lambda^2) = R_{0}(-\lambda^2)\,\one_{\supp \mathcal{X}}\,\mathcal{X}\, \one_{\supp \mathcal{X}} R_{0}(-\lambda^2) \in \mathcal{L}^1 , 
	\end{equation}
for some $\lambda > 0$ large enough. In this connection, we notice that
	\begin{equation*}
		\lf\| R_{0}(-\lambda^2)\,\one_{\supp \mathcal{X}} \ri\|_{\mathcal{L}^2}^2 \!
		\leqslant \lf\| R_0 \ri\|_2^2\, \lf| \supp \mathcal{X} \ri| < + \infty\,,
	\end{equation*}
	with a completely analogous bound for the adjoint operator. Hence, since $ \mathcal{X} $ is bounded, we conclude that \eqref{eq: R0XR0} holds by \cite[Thm. VI.22]{RS81}, which in turn implies \eqref{eq: RAchiS1}, whence the thesis.
\end{proof}

We are now able to prove the characterization of the spectrum of $ H_N^{(\mathrm{F})} $ and of its scattering w.r.t. the free Laplacian.

\begin{proof}[Proof of \cref{pro: spectrum F}]
		Consider any auxiliary operator $H_{\mathbf{A}}^{(\mathrm{F})}$, with $\mathbf{A}$ fulfilling \eqref{eq:AAasy} and \eqref{eq:AACou}. By \cref{lemma: RAR0compact} and \cref{lemma: RNRAS2} we readily infer that
			\begin{equation*}
				R_{N}^{(\mathrm{F})}(z) - R_{0}(z) = R_{N}^{(\mathrm{F})}(z) - R^{(\mathrm{F})}_{\mathbf{A}}(z) + R^{(\mathrm{F})}_{\mathbf{A}}(z) -  R_{0}(z) \in \mathcal{L}^{\infty}\,, \qquad \forall\, z \in \mathbb{C} \setminus [0,+\infty)\,.
			\end{equation*}
		Then, recalling that $H_{N}^{(\mathrm{F})}$ is non-negative, the thesis follows by Weyl's criterion.
		\end{proof}
	
\begin{proof}[Proof of \cref{thm: OmHNH0}]
Let again $H_{\mathbf{A}}^{(\mathrm{F})}$ be any auxiliary operator, with $\mathbf{A}$ fulfilling \eqref{eq:AAasy} and \eqref{eq:AACou}. By \cref{lemma: RN2RA2S1} and the invariance principle \cite[Corollary 3 (vol. III, p. 30)]{RS81}, we deduce the existence and completeness of the wave operators
	\begin{equation*}
		\Omega_{\pm}\lf(H_{N}^{(\mathrm{F})},H_{\mathbf{A}}^{(\mathrm{F})}\ri) := \slim_{t \to \mp \infty} e^{i t H_{N}^{(\mathrm{F})}} e^{-i t H_{\mathbf{A}}^{(\mathrm{F})}} .
	\end{equation*}
Then, on account of \cref{prop: RAR0scatt}, the thesis follows by the chain rule \cite[Proposition 2 (vol. III, p. 18)]{RS81}.
\end{proof}

\subsection{Singular perturbations: quadratic forms}\label{sec:SingPert}
For later reference, we first observe that the sesquilinear form defined by polarization starting from $Q_{N}^{(B)}$ w.r.t. the decompositions
\begin{equation*}
	\psi_j = \phi_{j,\lambda} +  \tx\sum_{n} \, e^{-i \mathbf{S}_{n}(\mathbf{x}_n) \cdot (\mathbf{x} - \mathbf{x}_n)}\, \xi_{n}\, \tx\sum_{\ell_n} \, q^{(\ell_n)}_{j,n}\, G^{(\ell_n)}_{\lambda,n} \qquad (j = 1,2 )
\end{equation*}
is given by
	\bml{
		Q^{(B)}_{N}[\psi_1,\psi_2] 
		= Q^{(\mathrm{F})}_{N}[\phi_{1,\lambda},\phi_{2,\lambda}] 
			- \lambda^2\, \langle \psi_1\,|\,\psi_2\rangle + \lambda^2\,\langle \phi_{1,\lambda}\,|\,\phi_{2,\lambda}\rangle \\
			+ \tx\sum_{n,\ell_n} \big({q_{1,n}^{(\ell_n)}}\big)^* \bigg( 2 \braketl{\,e^{-i \mathbf{S}_{n}(\mathbf{x}_n) \cdot (\mathbf{x} - \mathbf{x}_n)} \big(\check{\mathbf{S}}_{n}\,\xi_{n} - i \nabla \xi_{n} \big) G^{(\ell_n)}_{\lambda,n}}{ \lf(-i \nabla \!+\! \mathbf{A}_{n}\ri)\phi_{2,\lambda}}\\
				+ \braketl{e^{-i \mathbf{S}_{n}(\mathbf{x}_n) \cdot (\mathbf{x} - \mathbf{x}_n)}\lf( \check{\mathbf{S}}_{n}^2\, \xi_{n} + 2\, \mathbf{S}_n(\mathbf{x}_n) \cdot (\check{\mathbf{S}}_n \xi_{n} - i \nabla \xi_{n}) + \Delta \xi_{n} \ri)\! G^{(\ell_n)}_{\lambda,n}}{ \phi_{2,\lambda}} \bigg) \\
			+ \tx\sum_{n,\ell_n} q^{(\ell_n)}_{2,n} \bigg( 2 \braketr{ \lf(-i \nabla \!+\! \mathbf{A}_{n}\ri)\phi_{1,\lambda}}{\,e^{-i \mathbf{S}_{n}(\mathbf{x}_n) \cdot (\mathbf{x} - \mathbf{x}_n)} \big(\check{\mathbf{S}}_{n}\,\xi_{n} - i \nabla \xi_{n} \big) G^{(\ell_n)}_{\lambda,n}}\\
				+ \braketr{ \phi_{1,\lambda}}{e^{-i \mathbf{S}_{n}(\mathbf{x}_n) \cdot (\mathbf{x} - \mathbf{x}_n)}\lf( \check{\mathbf{S}}_{n}^2\, \xi_{n} + 2\, \mathbf{S}_n(\mathbf{x}_n) \cdot (\check{\mathbf{S}}_n \xi_{n} - i \nabla \xi_{n}) + \Delta \xi_{n} \ri)\! G^{(\ell_n)}_{\lambda,n}} \bigg) \\
			+ \tx\sum_{m, n, \ell_m, \ell_n^{\prime}} \big({q^{(\ell_m)}_{1,m}}\big)^* q^{(\ell_n^{\prime})}_{2,n} \left[B^{(\ell_m \ell_n^{\prime})}_{m\,n} + \delta_{m n} \lf( \tfrac{\pi\,\lambda^{2|\ell_m + \alpha_m|}}{2 \sin(\pi \alpha_m)}\,\delta_{\ell_m \ell_n^{\prime}} + \Xi^{(\ell_m \ell_n^{\prime})}_{n}(\lambda) \ri) \right]. \label{eq: defQ2}
	}

\begin{proof}[Proof of \cref{thm: Qbeta}]	\mbox{}

{\it i)} Let us first prove that the quadratic form \eqref{eq: Qbeta} is independent of $ \lambda > 0 $. To this avail, we fix $\lambda_1 \neq \lambda_2$ and consider, for any $\psi \in \dom\big[Q^{(B)}_{N}\big]$, the alternative representations
	\begin{equation*}
		\psi = \phi_{\lambda_j} \!+  \sum_{n,\ell_n}  e^{-i \mathbf{S}_{n}(\mathbf{x}_n) \cdot (\mathbf{x} - \mathbf{x}_n)}\,\xi_{n} q^{(\ell_n)}_{n} G^{(\ell_n)}_{\lambda_j,n} \qquad (j =1,2 )\,.
	\end{equation*} 
	Since $\xi_{n}(G^{(\ell_n)}_{\lambda_2,n} - G^{(\ell_n)}_{\lambda_1,n}) \in \dom\big[Q_{N}^{(\mathrm{F})}\big]$, for all $ n \in \{1,\dots,N\}$, we deduce that
	\begin{equation*}
		\phi_{\lambda_1} = \phi_{\lambda_2} \!+ \mbox{$\sum_{n,\ell_n}$} e^{-i \mathbf{S}_{n}(\mathbf{x}_n) \cdot (\mathbf{x} - \mathbf{x}_n)}\,\xi_{n} q^{(\ell_n)}_{n} \big( G^{(\ell_n)}_{\lambda_2,n} - G^{(\ell_n)}_{\lambda_1,n}\big)
	\end{equation*}
	and thus the ``charges'' $q^{(\ell_n)}_{n}$ are independent of $\lambda$. Furthermore, by using \eqref{eq: greeneq}, the identity $\mathbf{A}_n \cdot (\mathbf{x} - \mathbf{x}_n) = 0$ and the fact that $\xi_{n}$ is real-valued, we obtain, denoting for short by $ \psi_j $, $ j = 1,2 $, the two decompositions,
	\bml{\label{eq: l1l2}
		Q^{(B)}_{N}\lf[ \psi_1\ri] - Q^{(B)}_{N}\lf[ \psi_2 \ri] 
			= \sum_{n,\ell_n,\ell_n^{\prime}} \big({q^{(\ell_n)}_{n}}\big)^* q^{(\ell_n^{\prime})}_{n} \bigg[	
				\tfrac{\pi}{2\sin(\pi \alpha_n) }\,\big(\lambda_1^{2|\ell_n+\alpha_n|} - \lambda_2^{2|\ell_n+\alpha_n|}\big)\, \delta_{\ell_n \ell_n^{\prime}}  \\
				+ \lim_{r \to 0^{+}} \int_{\partial B_r(\mathbf{x}_n)}\hspace{-0.4cm} \diff \Sigma_n\; \xi_{n}^2 \left({\big(\partial_r G^{(\ell_n)}_{\lambda_1,n}\big)}^{\!*} G^{(\ell_n^{\prime})}_{\lambda_2,n} - \big({G^{(\ell_m)}_{\lambda_1,n}}\big)^{*} \big(\partial_r G^{(\ell_n^{\prime})}_{\lambda_2,n}\big) \right) \\
				+ \braketr{G^{(\ell_n)}_{\lambda_1,n}}{\xi_{n}( - i \nabla \xi_{n}) \cdot (-i \nabla + \mathbf{A}_n)\, G^{(\ell_n^{\prime})}_{\lambda_1,n}} - \braketr{\xi_{n} (- i \nabla \xi_{n}) \cdot (-i \nabla + \mathbf{A}_{n})\,G^{(\ell_n)}_{\lambda_1,n}}{ G^{(\ell_n^{\prime})}_{\lambda_1,n}} \\
				+ \braketr{G^{(\ell_n)}_{\lambda_2,n}}{\xi_{n}( - i \nabla \xi_{n}) \cdot (-i \nabla + \mathbf{A}_n)\, G^{(\ell_n^{\prime})}_{\lambda_2,n}} - \braketr{\xi_{n} (- i \nabla \xi_{n}) \cdot (-i \nabla + \mathbf{A}_{n})\,G^{(\ell_n)}_{\lambda_2,n}}{ G^{(\ell_n^{\prime})}_{\lambda_2,n}}
			\bigg]\,,
	}
where we dropped all the vanishing boundary terms. Since  $\xi_{n} = 1$ in an open neighborhood of $\mathbf{x}_n$ and thanks to \eqref{eq: checkSnLip}, \eqref{eq: limdomQF} and the asymptotic expansion \eqref{eq: G2asy0}, we indeed have 
	\bmln{
		\left|\int_{\partial B_r(\mathbf{x}_n)}\hspace{-0.4cm} \diff\Sigma_n\; \phi_{\lambda_2}^* \;{e^{-i \mathbf{S}_{n}(\mathbf{x}_n) \cdot (\mathbf{x} - \mathbf{x}_n)}\, \Big(i \big(\check{\mathbf{S}}_n \!\cdot\! \hat{\mathbf{r}}\big)\xi_{n} + (\partial_r \xi_{n}) - \xi_{n} \partial_r\Big)\Big( G^{(\ell_n^{\prime})}_{\lambda_2,n} - G^{(\ell_n^{\prime})}_{\lambda_1,n}\Big)} \right| \\
		\leqslant 2 \pi r\, \sqrt{\avg{\left|\phi_{\lambda_2}\right|^2}_{n}}\left( \sqrt{\avg{\big|\check{\mathbf{S}}_n \big|^2\, \Big| G^{(\ell_n^{\prime})}_{\lambda_2,n} - G^{(\ell_n^{\prime})}_{\lambda_1,n}\Big|^2}_{n}} + \sqrt{\avg{\Big|\partial_r\Big( G^{(\ell_n^{\prime})}_{\lambda_2,n} - G^{(\ell_n^{\prime})}_{\lambda_1,n}\Big)\Big|^2}_{n}}\right) \\
		\leqslant C \, \sqrt{\avg{\left|\phi_{\lambda_2}\right|^2}_{n}}\left(r^{2+|\ell_n^{\prime} + \alpha_n|} + r^{|\ell_n^{\prime} + \alpha_n|}\right) \xrightarrow[r \to 0^+]{} 0\,.
	}
By similar arguments we deduce the following relations:
	\begin{equation*}
		\bigg|\int_{\partial B_r(\mathbf{x}_n)}\hspace{-0.5cm} \diff\Sigma_n \; \xi_{n}^2 \big(\check{\mathbf{S}}_n \!\cdot \hat{\mathbf{r}}\big)\; {\big( G^{(\ell_n)}_{\lambda_2,n}\! -\! G^{(\ell_n)}_{\lambda_1,n}\big)}^{\!*} \, \big( G^{(\ell_n^{\prime})}_{\lambda_2,n} \!+ G^{(\ell_n^{\prime})}_{\lambda_1,n}\big) \bigg| 
		\leqslant C\, r^{2+|\ell_n+\alpha_n|\,-\,|\ell_n^{\prime} + \alpha_n|} \xrightarrow[r \to 0^+]{} 0\,;
	\end{equation*}
	\begin{equation*}
		\left|\int_{\partial B_r(\mathbf{x}_n)}\hspace{-0.4cm} \diff\mathbf{x}\, \big(\xi_{n} \partial_r \xi_{n}\big)\! \left( {\big( G^{(\ell_n)}_{\lambda_2,n} - G^{(\ell_n)}_{\lambda_1,n}\big)}^{\!*} \, G^{(\ell_n^{\prime})}_{\lambda_1,n} + \big({G^{(\ell_n)}_{\lambda_1,n}}\big)^*\, \big( G^{(\ell_n^{\prime})}_{\lambda_2,n} - G^{(\ell_n^{\prime})}_{\lambda_1,n} \big) \right) \right| \xrightarrow[r \to 0^+]{} 0\,;
	\end{equation*}
	\begin{equation*}
		\bigg|\int_{\partial B_r(\mathbf{x}_n)}\hspace{-0.5cm} d\Sigma_n\; \xi_{n}^2\, {\big( G^{(\ell_n)}_{\lambda_2,n} - G^{(\ell_n)}_{\lambda_1,n}\big)}^{\!*} \, \partial_r \big( G^{(\ell_n^{\prime})}_{\lambda_2,n} - G^{(\ell_n^{\prime})}_{\lambda_1,n}\big) \bigg| 
			\leqslant C\,r^{|\ell_n+\alpha_n|\,+\,|\ell_n^{\prime} + \alpha_n|} \xrightarrow[r \to 0^+]{} 0\,.
	\end{equation*}
On the other hand, using again the asymptotic expansion \eqref{eq: G2asy0} for $G^{(\ell_n)}_{\lambda,n}$, we get
	\bmln{
		\int_{\partial B_r(\mathbf{x}_n)}\hspace{-0.4cm} \diff\Sigma_n\; \xi_{n}^2 \left({\big(\partial_r G^{(\ell_n)}_{\lambda_1,n}\big)}^{\!*} G^{(\ell_n^{\prime})}_{\lambda_2,n} - \big({G^{(\ell_n)}_{\lambda_1,n}}\big)^* \big(\partial_r G^{(\ell_n^{\prime})}_{\lambda_2,n}\big) \right) \\
		= \delta_{\ell_n \ell_n^{\prime}}\! \left[
			\tfrac{1}{2}\,|\ell_n \!+\! \alpha_n|\,\Gamma\big(|\ell_n \!+\! \alpha_n|\big) \Gamma\big(\!-|\ell_n \!+ \!\alpha_n|\big) \big(\lambda_1^{2 |\ell_n + \alpha_n|} - \lambda_2^{2 |\ell_n + \alpha_n|} \big)
			+ \mathcal{O}\!\left(r^{2 \min\{|\ell_n + \alpha_n|\,,\, 1\,-\,|\ell_n + \alpha_n|\}}\right)
			\right] .
}
Applying the identities for the Euler gamma function in \cite[Eqs. 5.5.1 and 5.5.3]{OLBC10}, we deduce that
	\begin{equation*}
		\tfrac{\pi(\lambda_1^{2|\ell_n+\alpha_n|} - \lambda_2^{2|\ell_n+\alpha_n|})}{2\sin(\pi \alpha_n) }\, \delta_{\ell_n \ell_n^{\prime}}  \\
		+ \lim_{r \to 0^{+}} \int_{\partial B_r(\mathbf{x}_n)}\hspace{-0.4cm} \diff \Sigma_n\; \xi_{n}^2 \left({\big(\partial_r G^{(\ell_n)}_{\lambda_1,n}\big)}^{\!*} G^{(\ell_n^{\prime})}_{\lambda_2,n} - \big({G^{(\ell_n)}_{\lambda_1,n}}\big)^* \big(\partial_r G^{(\ell_n^{\prime})}_{\lambda_2,n}\big) \right)
		= 0\,.
	\end{equation*}
In view of the above arguments, the identity \eqref{eq: l1l2} reduces to
	\bmln{
		Q^{(B)}_{N}\lf[ \psi_1 \ri] - Q^{(B)}_{N}\lf[ \psi_2 \ri]  
		= - 2i \,\Re\! \sum_{n,\ell_n,\ell_n^{\prime}}\! \big({q^{(\ell_n)}_{n}}\big)^* q^{(\ell_n^{\prime})}_{n} \bigg(\!
			\braketr{G^{(\ell_n)}_{\lambda_1,n}}{(\xi_{n}\nabla \xi_{n}) \cdot (-i \nabla \!+\! \mathbf{A}_n) G^{(\ell_n^{\prime})}_{\lambda_1,n}} \\
			\,+\, \braketr{G^{(\ell_n)}_{\lambda_2,n}}{(\xi_{n} \nabla \xi_{n}) \cdot (-i \nabla \!+\! \mathbf{A}_n) G^{(\ell_n^{\prime})}_{\lambda_2,n}}
			\!\bigg)\,,
	}
where we used that the functions $\xi_{n}$ are real-valued. Since the quadratic form $Q^{(B)}_{N}$ is real-valued as well, the latter relation can be fulfilled only if 
	\begin{equation*} 
		Q^{(B)}_{N}\lf[ \psi_1 \ri] = Q^{(B)}_{N}\lf[ \psi_2 \ri],
	\end{equation*}
	which ultimately proves that the quadratic form is independent of the spectral parameter $\lambda$.

Next, we show that the form does not depend on the cut-off functions. Let $ \lf(\xi_{1,n}\ri)_{n \in \{1, \ldots, N\}}$, $ \lf(\xi_{2,n} \ri)_{n \in \{1, \ldots, N\}} $ be two families of such functions. Correspondingly, for any $\psi \in \dom\big[Q^{(B)}_{N}\big]$, we have two alternative representations which we denote again for short $ \psi_j $, $ j = 1,2 $. Notice that $(\xi_{2,n} - \xi_{1,n}) G^{(\ell_n)}_{\lambda,n} \in \dom[Q_{N}^{(\mathrm{F})}]$, for all $n \in \{1,\dots,N\}$ and $\ell_{n}\in \{0,-1\}$, so that 
\begin{equation*}
	\phi_{1,\lambda} = \phi_{2,\lambda} + \tx\sum_{n,\ell_n} e^{-i \mathbf{S}_{n}(\mathbf{x}_n) \cdot (\mathbf{x} - \mathbf{x}_n)}\,\big(\chi_{2,n} - \chi_{1,n} \big) q^{(\ell_n)}_{n} G^{(\ell_n)}_{\lambda,n}\,.
\end{equation*}
Using \eqref{eq: greeneq}, by direct computations we infer
	\bmln{
		Q^{(B)}_{N}\!\lf[ \psi_1 \ri] - Q^{(B)}_{N}\!\lf[ \psi_2 \ri] \\ = - 2i\, \Re\! \tx\sum_{n,\ell_n,\ell_n^{\prime}}\! \big({q^{(\ell_n)}_{n}}\big)^* q^{(\ell_n^{\prime})}_{n} \braketr{\big(\chi_{2,n} \!-\! \chi_{1,n} \big) G^{(\ell_n)}_{\lambda,n}}{\big( \nabla\chi_{2,n} \! - \! \nabla\chi_{1,n}  \big) \cdot (- i \nabla + \mathbf{A}_n) G^{(\ell_n^{\prime})}_{\lambda,n}},
	}
	where all the boundary terms coming from the integration by parts vanish since $\xi_{1,n} = \xi_{2,n} = 1$ in an open neighborhood of $\mathbf{x}_n$. Recalling again that the quadratic form is real, it appears that the latter relation can be fulfilled if and only if $ Q^{(B)}_{N}\!\lf[ \psi_1 \ri] = Q^{(B)}_{N}\!\lf[ \psi_2 \ri] $, which proves that the form is independent of the family of cut-off functions.

{\it ii).} To begin with, let us show that $Q^{(B)}_{N}$ is bounded from below.
In the sequel we set $ \one_{n} : = \one_{\supp \xi_{n}} $ and pick $\lambda > 0$ large enough. By Cauchy inequality, for any $ \epsilon \in (0,1) $ we infer
	\bml{\label{eq: boundedness proof 1}
Q_{N}^{(\mathrm{F})}[\phi_{\lambda}] 
		\geqslant \tx{1 \over 2}\, Q_{N}^{(\mathrm{F})}[\phi_{\lambda}] 
+ \tx{1 \over 2}\, \sum_{n = 1}^{N} \lf \| \one_{n} \lf(-i\nabla + \mathbf{A}_{n} + \mathbf{S}_{n} \ri)\phi_{\lambda} \ri\|_{2}^{2} \\
		\geqslant \tx{1 \over 2}\, Q_{N}^{(\mathrm{F})}[\phi_{\lambda}] 
			+ \tx\frac{1 - \epsilon}{2}\, \sum_{n = 1}^{N} \lf \| \one_{n} \lf(-i\nabla + \mathbf{A}_{n} \ri)\phi_{\lambda} \ri\|_{2}^{2}
			- \tx\frac{1 - \epsilon}{2\epsilon}\, \sum_{n = 1}^{N} \lf\| \one_{n} \mathbf{S}_{n} \ri\|_{\infty}^2\,\|\phi_{\lambda}\|_{2}^{2}\,.
	}
Furthermore, for any $ \epsilon_1, \epsilon_2 \in (0,1) $ we have
	\bml{\label{eq: boundedness proof 2}
	 	\Re \lf[ q^{(\ell_n)}_{n} \braketr{ \lf(-i \nabla + \mathbf{A}_{n} \ri)\phi_{\lambda}}{\,\,e^{-i \mathbf{S}_{n}(\mathbf{x}_n) \cdot (\mathbf{x} - \mathbf{x}_n)} \big(\check{\mathbf{S}}_{n}\,\xi_{n} - i \nabla \xi_{n} \big) G^{(\ell_n)}_{\lambda,n}} \ri] \\
		\geqslant -\, \tx{\epsilon_1 \over 2}\, \big\| \one_{n}  (-i \nabla + \mathbf{A}_{n}) \phi_{\lambda} \big\|_{2}^{2}
			- \tx{1 \over 2\epsilon_1} \lf|q_{n}^{(\ell_n)}\ri|^2\, \lf\| \big( \check{\mathbf{S}}_{n} \,\xi_{n} - i\nabla \xi_{n} \big) G^{(\ell_n)}_{\lambda,n} \ri\|_{2}^{2} \\
		\geqslant -\, \tx{\epsilon_1 \over 2}\, \lf\| \one_{n} \lf(-i \nabla + \mathbf{A}_{n}\ri) \phi_{\lambda} \ri\|_{2}^{2}
			- \tx{C \over \epsilon_1} \lf|q^{(\ell_n)}_{n}\ri|^2\, \lambda^{2|\ell_n + \alpha_{n}|-2}\,;
	}
	\bml{\label{eq: boundedness proof 3}
		\Re \lf[ q^{(\ell_n)}_{n} \braketr{ \phi_{\lambda}}{e^{-i \mathbf{S}_{n}(\mathbf{x}_n) \cdot (\mathbf{x} - \mathbf{x}_n)}\lf( \check{\mathbf{S}}_{n}^2\, \xi_{n} + 2 \mathbf{S}_n(\mathbf{x}_n) \cdot (\check{\mathbf{S}}_n \xi_{n} - i \nabla \xi_{n}) + \Delta \xi_{n} \ri)\! G^{(\ell_n)}_{\lambda,n}} \ri] \\
			\geqslant - \,\tx{\epsilon_2 \over 2}\,\| \one_{n} \phi_{\lambda}\|_{2}^{2}
- \tx{C \over \epsilon_2} \lf|q^{(\ell_n)}_{n}\ri|^2\, \lambda^{2|\ell_n + \alpha_{n}|-2}\,.
	}
Finally, let us consider the term $\Xi_{\alpha,n}(\lambda)$ defined in \eqref{eq: defXi}. Notice that \eqref{eq: checkSnLip} ensues $\check{\mathbf{S}}_{n}\!\cdot\!\mathbf{A}_{n}\,\xi_{n}^2 \in L^{\infty}(\mathbb{R}^2)$ and $\xi_{n} \check{\mathbf{S}}_n\,|\mathbf{x} - \mathbf{x}_{n}|^{-1} \in L^{\infty}(\mathbb{R}^2)$. Taking this into account and building on the basic inequality 
	\begin{equation*}
		\lf||\mathbf{x} - \mathbf{x}_{n}|\, (-i \nabla) \big(\xi_{n}  G^{(\ell_n)}_{\lambda,n}\big)\ri| \leqslant C \lf| G^{(\ell_n)}_{\lambda,n} \ri|,
	\end{equation*}
we get
	\beq\label{eq: boundedness proof 4}
		\lf|\Xi^{(\ell_n\ell_n^{\prime})}_{n}(\lambda)\ri|
		\leqslant C \lf\|G^{(\ell_n)}_{\lambda,n}\ri\|_{2} \lf\|G^{(\ell_n^{\prime})}_{\lambda,n} \ri\|_{2}	
		\leqslant C\, \lambda^{|\ell_n+\alpha_n| + |\ell_n^{\prime} + \alpha_n| - 2}\,.
	\eeq
For any $\lambda > 0$ sufficiently large, combining \eqref{eq: boundedness proof 1} -- \eqref{eq: boundedness proof 4}, we estimate
	\bmln{
		Q^{(B)}_{N}[\psi] 
		\geqslant \tx{1 \over 2}\, Q_{N}^{(\mathrm{F})}[\phi_{\lambda}] - \lambda^2 \lf\| \psi \ri\|_{2}^2 
			+ \tfrac{1 - \epsilon - 8\epsilon_1}{2}\, \tx\sum_{n } \lf \| \one_{n} \lf(-i\nabla + \mathbf{A}_{n} \ri)\phi_{\lambda} \ri\|_{2}^{2}
			+ \lf( \lambda^2 - C \tfrac{1 - \epsilon}{2\epsilon}  - 2\epsilon_2 \ri) \lf\| \phi_{\lambda} \ri\|_{2}^2 \\
			+ \lf[ \min_{n}\!\lf(\tx{\pi^2 \over \sin(\pi \alpha_n)}\ri) \lambda^{2 \min_{n,\ell_n} |\ell_n + \alpha_n|} - 2C \lf(1\!+ \tfrac{2}{\epsilon_1} \! +\tfrac{1}{\epsilon_2} \ri)\,\lambda^{-2(1-\max_{n,\ell_n} |\ell_n + \alpha_n|)} - \lf\|B \ri\|_{\infty} \ri] \lf| \mathbf{q} \ri|^2 ,
	}
where $ \lf\| B \ri\|_{\infty} := \max_{m,n,\ell_m \ell_n^{\prime}}\!  |B^{(\ell_m \ell_n^{\prime})}_{m\,n} |$.
Now, by the positivity of $Q_{N}^{(\mathrm{F})}$, if we take $  \epsilon, \epsilon_2 \in (0,1) $, $0 < \epsilon_1 < (1 - \epsilon)/8 $ and $\lambda > 0$ large enough (notice that $\max_{n,\ell_n} |\ell_n + \alpha_n| < 1$), the above relation implies that $Q^{(B)}_{N}[\psi] \geqslant - \lambda^2\, \|\psi\|_{2}^{2} $, {\it i.e.}, the form is bounded from below. More precisely, for any $ \lambda $ large enough, we deduce that there exists $c_{\lambda} > 0 $ such that
	\begin{equation*}
	Q^{(B)}_{N}[\psi] + \lambda^2\,\|\psi\|_{2}^{2} \geqslant c_{\lambda} \lf[Q_{N}^{(\mathrm{F})}[\phi_{\lambda}] +  \|\phi_{\lambda}\|_{2}^{2} +    \lf| \mathbf{q} \ri|^2 \ri] ,
	\end{equation*}
{\it i.e.}, the quadratic form is also coercive. This allows to infer that the form $Q^{(B)}_{N}$ is closed by classical arguments (see, {\it e.g.}, \cite{Te90} and \cite[Proof of Thm. 2.4]{CO18}).
\end{proof}

Hereafter we prove the first part of \cref{cor:domHbeta}, regarding the characterization of the domain and action of the self-adjoint operators $H_{N}^{(B)}$, $B \in M_{2N,\,\mathrm{Herm}}(\mathbb{C})$, corresponding to the quadratic forms $Q_{N}^{(B)}$.

\begin{proof}[Proof of \cref{cor:domHbeta} - part I]
Consider the sesquilinear form \eqref{eq: defQ2} associated to $Q_{N}^{(B)}$. Let us first assume $ \mathbf{q}_1 = \mathbf{0} $, {\it i.e.}, $\psi_{1} = \phi_{1,\lambda}$; then, upon integration by parts, we obtain
	\bmln{
		Q^{(B)}_{N}[\phi_{1,\lambda},\psi_2] 
		= \meanlrlr{\phi_{1,\lambda}}{H_{N}}{\phi_{2,\lambda}}
			+ \sum_{n,\ell_n} q^{(\ell_n)}_{2,n} \bigg[2 \braketr{\phi_{1,\lambda}}{\,e^{-i \mathbf{S}_{n}(\mathbf{x}_n) \cdot (\mathbf{x} - \mathbf{x}_n)} \big(\check{\mathbf{S}}_{n}\,\xi_{n} - i \nabla \xi_{n} \big) \lf(-i \nabla \!+\! \mathbf{A}_{n}\ri)G^{(\ell_n)}_{\lambda,n}} \\
			+  \braketr{\phi_{1,\lambda}}{\,e^{-i \mathbf{S}_{n}(\mathbf{x}_n) \cdot (\mathbf{x} - \mathbf{x}_n)} \lf( (\check{\mathbf{S}}_{n}^2 - \lambda^2)\, \xi_{n} + 2\check{\mathbf{S}}_{n} \!\cdot\! (-i \nabla \xi_{n}) - \Delta \xi_{n} \ri) G^{(\ell_n)}_{\lambda,n}} \bigg]\,, 
	}
	where again the boundary terms have been dropped, since by the asymptotic relations in Eq.\,\eqref{eq: limdomQF} of \cref{prop:QF} and arguments similar to those outlined in the proof of \cref{thm: Qbeta}, we get
	\begin{equation*}
		\lf| \int_{\partial B_r(\mathbf{x}_n)}\hspace{-0.6cm} \diff\Sigma_n\; \phi_{1,\lambda}^* \,\lf[ \partial_r\phi_{2,\lambda} \!+\! i \lf(\mathbf{S}_n \!\cdot \hat{\mathbf{r}} \ri)\phi_{2,\lambda} \ri] \ri|
		\leqslant C \,r\, \sqrt{\avg{|\phi_1|^2}_{n}} \lf[ \sqrt{\avg{\lf| \partial_r \phi_2 \ri|^2}_{n}}\!
		+  \,\sqrt{\avg{|\phi_2|^2}_{n}} \,\ri] \!\xrightarrow[r \to 0^+]{} \,0\,;
	\end{equation*}
	\begin{equation*}
		\lf| \int_{\partial B_r(\mathbf{x}_n)}\hspace{-0.4cm} \diff \Sigma_n\; \phi_{1,\lambda}^*\; e^{-i \mathbf{S}_{n}(\mathbf{x}_n) \cdot (\mathbf{x} - \mathbf{x}_n)} \lf[ i \lf(\check{\mathbf{S}}_{n}\!\cdot \hat{\mathbf{r}} \ri)\xi_{n} + (\partial_r\xi_{n}) \ri] G^{(\ell_n)}_{\lambda,n} \ri|
		\leqslant C\, r^{2-|\ell_n + \alpha_n|}\, \sqrt{\avg{|\phi_1|^2}_{n}} \xrightarrow[r \to 0^+]{} \;0\,.
	\end{equation*}
Thus, the condition $Q^{(B)}_{N}[\phi_{1,\lambda},\psi_2] = \langle \phi_{1,\lambda}\,|\,w\rangle$ for some $w = : H_{N}^{(B)}\, \psi_2 \!\in\! L^2(\mathbb{R}^2)$ can be fulfilled only if $ H_{N}\, \phi_{2,\lambda} \in L^2(\mathbb{R}^2)$, {\it i.e.}, $ \phi_{2,\lambda} \in \dom\big(H_{N}^{(\mathrm{F})}\big)$, and 
	\bml{\label{defw1}
		w = H_{N}\, \phi_{2,\lambda}
			+ \tx\sum_{n,\ell_n} q^{(\ell_n)}_{2,n}\, e^{-i \mathbf{S}_{n}(\mathbf{x}_n) \cdot (\mathbf{x} - \mathbf{x}_n)} \lf[  2\big(\check{\mathbf{S}}_{n}\,\xi_{n} - i \nabla \xi_{n} \big) \lf(-i \nabla \!+\! \mathbf{A}_{n}\ri)G^{(\ell_n)}_{\lambda,n} \ri. \\
				\lf. + \lf( (\check{\mathbf{S}}_{n}^2 - \lambda^2)\, \xi_{n} + 2\check{\mathbf{S}}_{n} \!\cdot\! (-i \nabla \xi_{n}) - \Delta \xi_{n} \ri) G^{(\ell_n)}_{\lambda,n} \ri] .
	}
	
Let now assume that $ \mathbf{q} \neq \mathbf{0} $. Integrating by parts and checking again that the boundary contributions vanish, we get
	\bmln{
		Q^{(B)}_{N}[\psi_1,\psi_2] 
		= Q^{(B)}_{N}[\phi_{1,\lambda},\psi_2] \\
		+ \tx\sum_{m,\ell_m} \big(q_{1,m}^{(\ell_m)}\big)^* \lf( 2 \braketl{\,e^{-i \mathbf{S}_{m}(\mathbf{x}_m) \cdot (\mathbf{x} - \mathbf{x}_m)} \big(\check{\mathbf{S}}_{m}\,\xi_{m} - i \nabla \xi_{m} \big) G^{(\ell_m)}_{\lambda,m}}{ \lf(-i \nabla \!+\! \mathbf{A}_{m}\ri)\phi_{2,\lambda}} \ri. \\
			\lf. + \braketl{e^{-i \mathbf{S}_{m}(\mathbf{x}_m) \cdot (\mathbf{x} - \mathbf{x}_m)}\!\lf( (\check{\mathbf{S}}_{m}^2\! - \lambda^2)\, \xi_{m} + 2\, \mathbf{S}_{m}(\mathbf{x}_m) \cdot (\check{\mathbf{S}}_m \xi_{m} - i \nabla \xi_{m}) + \Delta \xi_{m} \ri)\! G^{(\ell_m)}_{\lambda,m}}{ \phi_{2,\lambda}} \ri) \\
		+ \tx\sum_{m, n,\ell_m,\ell_n^{\prime}} \big({q^{(\ell_m)}_{1,m}}\big)^* q^{(\ell_n^{\prime})}_{2,n}\! \left[B^{(\ell_m \ell_n^{\prime})}_{m\,n} + \delta_{mn} \lf( \tfrac{\pi\,\lambda^{2|\ell_n + \alpha_n|}}{2\sin(\pi \alpha_n)}\,\delta_{\ell_n \ell_n^{\prime}}\! + \Xi^{(\ell_n \ell_n^{\prime})}_{n}(\lambda) - \lambda^2 \braketr{\xi_{n} G^{(\ell_n)}_{\lambda,n}}{\xi_{n} G^{(\ell_n^{\prime})}_{\lambda,n}}\! \ri) \right].
	}
Then, demanding $Q^{(B)}_{N}[\psi_1,\psi_2] = \langle \psi_1\,|\,w\rangle$,  with $w$ as in \eqref{defw1}, implies
	\bml{\label{bcid}
		\tx\sum_{m, n,\ell_m,\ell_n^{\prime}} \big({q^{(\ell_m)}_{1,m}}\big)^* q^{(\ell_n^{\prime})}_{2,n} \lf[ B^{(\ell_m \ell_n^{\prime})}_{m\,n} + \delta_{mn} \lf( \tfrac{\pi\,\lambda^{2|\ell_n + \alpha_n|}}{2\sin(\pi \alpha_n)}\,\delta_{\ell_m \ell_n^{\prime}} \ri.\ri. \\
			\lf. \lf. - \braketr{G^{(\ell_m)}_{\lambda,n}}{\big[ (-i\nabla) \!\cdot\! \big(\xi_{n} (-i\nabla \xi_{n})\big) \big] G^{(\ell_n^{\prime})}_{\lambda,n} } 
			- \braketr{G^{(\ell_m)}_{\lambda,n}}{2 \xi_{n} (-i\nabla \xi_{n}) \!\cdot\! (- i\nabla \!+\! \mathbf{A}_{n}) G^{(\ell_n^{\prime})}_{\lambda,n}} \ri) \ri] \\
		= \tx\sum_{m,\ell_m} \big({q^{(\ell_m)}_{1,m}}\big)^* \braketr{G^{(\ell_m)}_{\lambda,m}}{ \lf[ (-i \nabla \!+\! \mathbf{A}_m)^2 + \lambda^2 \ri] e^{i \mathbf{S}_{m}(\mathbf{x}_m) \cdot (\mathbf{x} - \mathbf{x}_m)}\, \xi_{m} \phi_{2,\lambda}} .
	}
Integrating by parts twice and using \eqref{eq: greeneq}, we obtain
	\bmln{
		\braketr{G^{(\ell_m)}_{\lambda,n}}{\big[ (-i\nabla) \!\cdot\! \big(\xi_{n} (-i\nabla \xi_{n})\big) \big] G^{(\ell_n^{\prime})}_{\lambda,n} } 
			+ \braketr{G^{(\ell_m)}_{\lambda,n}}{2 \xi_{n} (-i\nabla \xi_{n}) \!\cdot\! (- i\nabla \!+\! \mathbf{A}_{n}) G^{(\ell_n^{\prime})}_{\lambda,n}} \\
		= \tfrac{1}{2} \lf[ \braketr{G^{(\ell_m)}_{\lambda,n}}{(-i\nabla \xi_{n}^2) \!\cdot\! (- i\nabla \!+\! \mathbf{A}_{n}) G^{(\ell_n^{\prime})}_{\lambda,n}} + \braketr{(- i\nabla \!+\! \mathbf{A}_{n})G^{(\ell_m)}_{\lambda,n}}{(- i\nabla \xi_{n}^2) G^{(\ell_n^{\prime})}_{\lambda,n}} \ri] \\
			+ \lim_{r \to 0^{+}} \int_{\partial B_r(\mathbf{x}_n)}\hspace{-0.5cm} \diff \Sigma_n\; \big({G^{(\ell_m)}_{\lambda,n}}\big)^{*}  (\xi_{n} \partial_r \xi_{n})\, G^{(\ell_n^{\prime})}_{\lambda,n} 
			= 0\,,
	}
where we used that $\xi_{n} = 1$ in an open neighborhood of $\mathbf{x}_{n}$  and \eqref{eq: G2exp}.
Concerning the r.h.s. of  \eqref{bcid}, by similar arguments and recalling again that $\mathbf{A}_{n} \cdot (\mathbf{x} - \mathbf{x}_{n}) = 0$, we infer that
	\bmln{
		\braketr{G^{(\ell_m)}_{\lambda,m}}{ \lf[ (-i \nabla \!+\! \mathbf{A}_m)^2 + \lambda^2 \ri] e^{i \mathbf{S}_{m}(\mathbf{x}_m) \cdot (\mathbf{x} - \mathbf{x}_m)}\, \xi_{m} \phi_{2,\lambda}} \\
		= \lim_{r \to 0^{+}} \int_{\mathbb{R}^2 \setminus B_r(\mathbf{x}_{m})}\hspace{-0.5cm} \diff\mathbf{x}\; \big({G^{(\ell_m)}_{\lambda,m}}\big)^*\, \lf[ (-i \nabla \!+\! \mathbf{A}_m)^2 + \lambda^2 \ri] \lf(e^{i \mathbf{S}_{m}(\mathbf{x}_m) \cdot (\mathbf{x} - \mathbf{x}_m)}\, \xi_{m} \phi_{2,\lambda}\ri) \\
		= \lim_{r \to 0^{+}} \int_{\partial B_r(\mathbf{x}_m)}\hspace{-0.5cm} \diff\Sigma_m \lf( \big({G^{(\ell_m)}_{\lambda,m}}\big)^{*} \partial_r \phi_{2,\lambda} - \big({\partial_r G^{(\ell_m)}_{\lambda,m}}\big)^{*} \phi_{2,\lambda}\ri) ,
	}
where the boundary term vanishes since
	\bmln{
		\lf| \int_{\partial B_r(\mathbf{x}_m)}\hspace{-0.5cm} \diff\Sigma_m\; e^{i \mathbf{S}_{m}(\mathbf{x}_m) \cdot (\mathbf{x} - \mathbf{x}_m)}\, \big(G^{(\ell_m)}_{\lambda,m}\big)^*\, \Big(i (\mathbf{S}_{m}(\mathbf{x}_m) \!\cdot \hat{\mathbf{r}})\, \xi_{m} + \partial_r \xi_{m}\Big)\, \phi_{2,\lambda} \ri| \\
		\leqslant 2 \pi r\, \lf\|\mathbf{S}_{m} \ri\|_{\infty}\; \sqrt{ \avg{\lf|G^{(\ell_m)}_{\lambda,m}\ri|^2}_{m}}\; \sqrt{\avg{\left|\phi_{2,\lambda}\right|^2}_{m}}
		\leqslant C\, r^{1 - |\ell_m + \alpha_m|}\, \sqrt{\avg{\left|\phi_{2,\lambda}\right|^2}_{m}} \; \xrightarrow[r \to 0^+]{}\; 0\,.
	}
Using the asymptotic expansion \eqref{eq: G2asy0} of $G^{(\ell_n)}_{\lambda,n}$, we then obtain
	\bmln{
		\braketr{G^{(\ell_m)}_{\lambda,m}}{ \lf[ (-i \nabla \!+\! \mathbf{A}_m)^2 + \lambda^2 \ri] e^{i \mathbf{S}_{m}(\mathbf{x}_m) \cdot (\mathbf{x} - \mathbf{x}_m)}\, \xi_{m} \phi_{2,\lambda}} \\
		= \tfrac{\Gamma\lf(|\ell_m + \alpha_m|\ri)}{ 2^{1 - |\ell_m + \alpha_m|}} \lim_{r \to 0^{+}} \tfrac{1}{r^{|\ell_m + \alpha_m|}} \int_{0}^{2\pi}\!\! \diff\theta\; \tfrac{e^{-i\, \ell_m \theta}}{\sqrt{2\pi}} \Big( r\,\partial_r \phi_{2,\lambda} + |\ell_m + \alpha_m|\, \phi_{2,\lambda} \Big) \,.
	}
Summing up, the above results entail
	\bmln{
		\sum_{n,\ell_n'} q^{(\ell_n')}_{2,n} \lf[B^{(\ell_m \ell_n')}_{m\,n} + \tfrac{\pi^2\,\lambda^{2|\ell_n + \alpha_n|}}{\sin(\pi \alpha_n)}\,\delta_{mn}\,\delta_{\ell_m \ell_n^{\prime}} \ri] \\
		= \lim_{r \to 0^{+}} \tfrac{\pi\, 2^{|\ell_m + \alpha_m|} \,\Gamma\big(|\ell_m + \alpha_m|\big)}{r^{|\ell_m + \alpha_m|}}  \avg{ \Big( r\,\partial_r \phi_{2,\lambda} + |\ell_m + \alpha_m|\, \phi_{2,\lambda} \Big) \tfrac{e^{-i\, \ell_m \theta}}{ \sqrt{2\pi}}}_{m} ,
	}
which completes the derivation of the domain \eqref{domHbe1}. 
\end{proof}

\subsection{Singular perturbations: Kre{\u\i}n theory}\label{subsec: Krein}

Building on the results derived in the previous section and using the general theory described in \cite{Po01}, we henceforth present an equivalent characterization of the Hamiltonian operators $H_{N}^{(B)}$ in terms of singular perturbations of the Friedrichs Hamiltonian $H_{N}^{(\mathrm{F})}$. Fixing arbitrarily $z_0 \in \mathbb{C} \setminus [0,+\infty)$, we recall \eqref{eq: defLambda}:
	\begin{equation*}
		\Lambda(z) := \bm{\tau} \lf( \tfrac{1}{2} \lf( \mathcal{G}(z_0) + \mathcal{G}(\bar{z}_0) \ri) - \mathcal{G}(z) \ri) : \mathbb{C}^{2N} \to \mathbb{C}^{2N}\,.
	\end{equation*}
Abstract resolvent arguments grant the well-posedness of $\Lambda(z)$ and further ensure the following, for all $z,w \in \mathbb{C} \setminus [0,+\infty)$ (see \cite[Lemma 2.2]{Po01} and the related discussion):
	\begin{equation}
		\Lambda(z) - \Lambda(w) = (z-w) \,\breve{\mathcal{G}}(w)\, \mathcal{G}(z)\,, \qquad 
		\big(\Lambda(z)\big)^{*} = \Lambda(z^{*})\,. \label{eq: Lambdaprop}
	\end{equation}
Keeping in mind that $H_{N}^{(\mathrm{F})}$ is bounded from below, for any Hermitian matrix $\Theta \in M_{2N,\,\mathrm{Herm}}(\mathbb{C})$ there certainly exists a non-empty open set $\mathcal{Z} \subset \mathbb{C} \setminus [0,+\infty)$ such that $\Theta + \Lambda(z)$ is invertible for all $z \in \mathcal{Z}$. Taking this into account, from \cite[Theorem 2.1]{Po01} we deduce the following.

	\begin{proposition}[Singular perturbations of $R_N^{(\mathrm{F})}$]
		\label{thm: extRF}
		\mbox{}		\\
		For any $z \in \mathcal{Z}$, the bounded operator 
			\begin{equation*}
				R_N^{(\Theta)}(z) := R_N^{(\mathrm{F})}(z) + \mathcal{G}(z) \big[ \Theta + \Lambda(z) \big]^{-1} \breve{\mathcal{G}}(z)\,,
			\end{equation*}
		is the resolvent of the self-adjoint operator $H_{N}^{(\Theta)}$ coinciding with $H_{N}^{(\mathrm{F})}$ on $ \ker \bm{\tau}$ and defined by \eqref{eq:HNTheta}, {\it i.e.},
		\begin{gather*}
			\dom\big(H_{N}^{(\Theta)}\big) = \lf\{ \psi \in L^2(\mathbb{R}^2)\;\Big|\; \psi = \varphi_{z} + \mathcal{G}(z) \mathbf{q},\; \varphi_{z} \in \dom\big(H_{N}^{(\mathrm{F})}\big),\; \mathbf{q} \!\in \mathbb{C}^{2N}, \; \bm{\tau} \varphi_z = \big[\Theta + \Lambda(z)\big] \mathbf{q} \ri\} , \nonumber \\
		\big(H_{N}^{(\Theta)} - z\big) \psi = \big(H_{N}^{(\mathrm{F})} - z\big) \varphi_z\,.
		\end{gather*}
	\end{proposition}

We next prove that the self-adjoint operators $H_{N}^{(\Theta)}$ defined in \cref{thm: extRF} are nothing but a reparametrization of the Hamiltonians $H_{N}^{(B)}$ discussed previously in \cref{cor:domHbeta}.

	\begin{proposition}
		\label{prop:BTheta}
		\mbox{}	\\
		For any $ \lambda_0 > 0 $ large enough, there is a one-to-one correspondence between the families of self-adjoint realizations $ \lf\{ H_{N}^{(\Theta)} \ri\}_{\Theta \in M_{2N,\,\mathrm{Herm}}(\mathbb{C})} $ and $ \lf\{ H_{N}^{(B)} \ri\}_{B \in M_{2N,\,\mathrm{Herm}}(\mathbb{C})} $, which is realized by
			\bml{\label{eq: TeB}
				\Theta_{mn}^{\ell_m \ell_n'} = B^{(\ell_m \ell_n')}_{m\,n} + \tfrac{\pi\,\lambda_0^{2|\ell_n + \alpha_n|}}{2 \sin(\pi \alpha_n)}\,\delta_{mn}\,\delta_{\ell_m \ell_n'} \\
					+ \braketr{e^{- i \mathbf{S}_{m}(\mathbf{x}_m) \cdot (\mathbf{x} - \mathbf{x}_m)} P_m\, G_{\lambda_0,m}^{(\ell_m)}}{ \big[1 \!+\! T^{*}_N\big(-\lambda_0^2\big)\big]^{-1}  e^{-i \mathbf{S}_{n}(\mathbf{x}_n) \cdot (\mathbf{x} - \mathbf{x}_n)} \xi_{n} G^{(\ell_n')}_{\lambda_0,n} } .
			}
	\end{proposition}

Before discussing the proof of the above result, we need to state some technical lemmas. Making reference to \cref{lemma: Tn}, we henceforth take  $\lambda, \lambda_0 > 0$ large enough so that
	\begin{equation}\label{eq: TN}
		\big\|T_N\big(\!-\lambda^2\big)\big\| < 1\,,		\qquad	\big\|T_N\big(\!-\lambda_0^2\big)\big\| < 1
	\end{equation}
This choice of $\lambda, \lambda_0$ ensures the well-posedness of the representation \eqref{eq: RNFried} for the Friedrichs resolvent in both $ z = - \lambda^2 $ and $ z_0 : = - \lambda_0^2 $.

	\begin{lemma}
		\label{lemma: Gzexp}
		\mbox{}		\\
		For any $\lambda > 0$ fulfilling \eqref{eq: TN} and for all $\psi \in L^2(\mathbb{R}^2)$, $ \mathbf{q} \in \mathbb{C}^{2N}$ there holds:
			\begin{align}
				\breve{\mathcal{G}}(-\lambda^2) \psi &= \tx\bigoplus_{n,\ell_n}\! \braketl{\big[1 \!+\! T^{*}_N\big(\!-\lambda^2\big)\big]^{-1} e^{-i \mathbf{S}_{n}(\mathbf{x}_n) \cdot (\mathbf{x} - \mathbf{x}_n)} \xi_{n} G^{(\ell_n)}_{\lambda,n} }{ \psi } \,; \label{eq: expGhat}\\
				\mathcal{G}(-\lambda^2) \mathbf{q} &= \big[1 \!+\! T^{*}_N\big(\!-\lambda^2\big)\big]^{-1} \tx\sum_{n,\ell_n} q_n^{(\ell_n)} e^{-i \mathbf{S}_{n}(\mathbf{x}_n) \cdot (\mathbf{x} - \mathbf{x}_n)} \xi_{n} G^{(\ell_n)}_{\lambda,n} \,. \label{eq: expG}
			\end{align}
		Moreover, for any fixed $\lambda_0 > 0$ fulfilling \eqref{eq: TN}, there holds
		\bml{
			\lf[\Lambda\big(\!-\lambda^2\big)\ri]_{mn}^{\ell_m \ell_n'}\!
			= \tfrac{\pi\,\lambda^{2|\ell_n + \alpha_n|}}{2 \sin(\pi \alpha_n)}\,\delta_{mn}\delta_{\ell_m \ell_n'}
				- \tfrac{\pi\,\lambda_0^{2|\ell_n + \alpha_n|}}{ 2 \sin(\pi \alpha_n)}\,\delta_{mn}\delta_{\ell_m \ell_n'} \\
				+ \braketr{e^{- i \mathbf{S}_{m}(\mathbf{x}_m) \cdot (\mathbf{x} - \mathbf{x}_m)} P_m\, G_{\lambda,m}^{(\ell_m)}}{ \big[1 \!+\! T^{*}_N\big(\!-\lambda^2\big)\big]^{-1}  e^{-i \mathbf{S}_{n}(\mathbf{x}_n) \cdot (\mathbf{x} - \mathbf{x}_n)} \xi_{n} G^{(\ell_n')}_{\lambda,n} } \\
				- \braketr{e^{- i \mathbf{S}_{m}(\mathbf{x}_m) \cdot (\mathbf{x} - \mathbf{x}_m)} P_m\, G_{\lambda_0,m}^{(\ell_m)}}{ \big[1 \!+\! T^{*}_N\big(\!-\lambda_0^2\big)\big]^{-1}  e^{-i \mathbf{S}_{n}(\mathbf{x}_n) \cdot (\mathbf{x} - \mathbf{x}_n)} \xi_{n} G^{(\ell_n')}_{\lambda_0,n} } . \label{eq: expLambda}
		}
	\end{lemma}

\begin{proof}
We first point out that the expressions in the second and third lines of \eqref{eq: expLambda} are well-defined since
	\begin{equation*}
		P_{n}\, G^{(\ell_n)}_{\lambda,n} \in L^2(\mathbb{R}^2)\,.
	\end{equation*}
Let us further highlight that the second identity in \eqref{eq: Lambdaprop} grants the self-adjointness of $\Lambda\big(\!-\lambda^2\big)$, meaning that
	\begin{equation*}
		\lf[\Lambda\big(\!-\lambda^2\big)\ri]_{mn}^{\ell_m \ell_n'} = \left({\lf[\Lambda\big(\!-\lambda^2\big)\ri]_{n m}^{\ell_n' \ell_m}}\right)^{\!*} .
	\end{equation*}
As a matter of fact, taking into account that $G^{(\ell_n)}_{\lambda,n} = \big(\tau^{(\ell_n)}_{n} R^{(\mathrm{F})}_{n}(-\lambda^2)\big)^{*}$, one gets that the matrix
	\begin{equation*}
		\left(\, \braketr{e^{- i \mathbf{S}_{m}(\mathbf{x}_m) \cdot (\mathbf{x} - \mathbf{x}_m)} P_m\, G_{\lambda,m}^{(\ell_m)}}{ \big[1 \!+\! T^{*}_N\big(\!-\lambda^2\big)\big]^{-1}  e^{-i \mathbf{S}_{n}(\mathbf{x}_n) \cdot (\mathbf{x} - \mathbf{x}_n)} \xi_{n} G^{(\ell_n')}_{\lambda,n} } \,\right)^{\ell_m \ell_n'}_{mn}
	\end{equation*}
is itself Hermitian.
Recalling the explicit expression \eqref{eq: RnAT} for $R^{(\mathrm{F})}_{n}(z)$,  we get for any $\psi \in L^2(\mathbb{R}^2)$
	\bml{
		\tau_{n}^{(\ell_n)} R^{(\mathrm{F})}_{n}(-\lambda^2) \psi 
		= 2^{|\ell_n + \alpha_n|-1} \,\Gamma\big(|\ell_n + \alpha_n|\big) \int_{0}^{+\infty}\hspace{-0.3cm} \diff r'\! \int_{0}^{2\pi} \!\!\! \diff\theta'\, \tfrac{e^{-i \ell_n \theta'}}{ \sqrt{2\pi}}\,\psi(r',\theta')\; \times \\
				\times\, \lim_{r \to 0^{+}} \tfrac{1}{r^{|\ell_n + \alpha_n|}} \lf[
				\one_{\{r < r'\}}\, \lf( |\ell_n + \alpha_n|\, I_{|\ell_n + \alpha_n|}(\lambda r) + r\,\partial_{r} I_{|\ell_n + \alpha_n|}(\lambda r)) \ri) \,K_{|\ell_n + \alpha_n|}(\lambda r')	\ri. \\
				\lf. + \one_{\{r > r'\}}\, I_{|\ell_n + \alpha_n|}(\lambda r') \lf( |\ell_n + \alpha_n|\,K_{|\ell_n + \alpha_n|}(\lambda r) + r\,\partial_{r} K_{|\ell_n + \alpha_n|}(\lambda r) \ri)
			\ri]  \\
		= \int_{0}^{+\infty}\hspace{-0.4cm} \diff r' \int_{0}^{2\pi} \!\!\! \diff\theta' \,\lambda^{|\ell_n + \alpha_n|} K_{|\ell_n + \alpha_n|}(\lambda r')\, \tfrac{e^{-i \ell_n \theta'}}{ \sqrt{2\pi}}\; \psi(r',\theta')
		= \int_{\mathbb{R}^2} \!\! \diff\mathbf{x}' \, \lf( G^{(\ell_n)}_{\lambda,n}(\mathbf{x}') \ri)^* \, \psi(\mathbf{x}')\,,
		\label{eq: tauRGn}
	}
where the last identity follows by comparison with the explicit expression \eqref{eq: G2exp} for the defect function. Indeed, here we used the Lebesgue dominated convergence theorem and the following asymptotic expansions for the Bessel functions $I_\nu(w),K_\nu(w)$ ($\nu > 0$ and $w> 0$) in the limit $w \to 0^{+}$ and analogous asymptotics for their derivatives \cite[\S 10.31]{OLBC10}:
	\begin{gather*}
		I_{\nu}(w) = w^{\nu} \lf[ \tfrac{1}{2^{\nu} \Gamma(1 + \nu)} + \tfrac{1}{2^{2 + \nu} \Gamma(2 + nu)}\,w^2 + \mathcal{O}(w^4) \ri] ,  \\
		K_{\nu}(w) = \tfrac{1}{w^{\nu}} \lf[ \tfrac{\Gamma(\nu)}{2^{1-\nu}} - \tfrac{\Gamma(-1 + \nu)}{ 2^{3-\nu}}\, w^2 + \mathcal{O}(w^4) \ri] + w^{\nu} \lf[ \tfrac{\Gamma(-\nu)}{2^{1+ \nu}} - \tfrac{\Gamma(-1-\nu)}{ 2^{3+\nu}}\, w^2 + \mathcal{O}(w^4) \ri] .
	\end{gather*}

	Using now the representation \eqref{eq: RNFried} for the Friedrichs resolvent, we obtain
	\bmln{
		\breve{\mathcal{G}}(-\lambda^2) \psi 
		=\tx \bigoplus_{n,\ell_n} \tau_{n}^{(\ell_n)} R^{(\mathrm{F})}_{n}(-\lambda^2)\, \xi_{n}\,e^{i \mathbf{S}_{n}(\mathbf{x}_n) \cdot (\mathbf{x} - \mathbf{x}_n)}\, \big[1 + T_N(-\lambda^2)\big]^{-1} \psi \\
		= \tx\bigoplus_{n,\ell_n}\! \braketr{e^{-i \mathbf{S}_{n}(\mathbf{x}_n) \cdot (\mathbf{x} - \mathbf{x}_n)} \xi_{n} G^{(\ell_n)}_{\lambda,n} }{\big[1 \!+\! T_N(-\lambda^2)\big]^{-1} \psi } ,
	}
which proves \eqref{eq: expGhat}	for all $\psi \in L^2(\mathbb{R}^2)$.
Taking this into account, for any $\mathbf{q} \in \mathbb{C}^{2N}$, we get
\bmln{	
	\braketl{\mathcal{G}(-\lambda^2) \mathbf{q}}{\psi} = \braketr{\mathbf{q}}{\breve{\mathcal{G}}(-\lambda^2) \psi }
		=  \tx\sum_{n,\ell_n}   \big(q_n^{(\ell_n)}\big)^* \braketr{e^{-i \mathbf{S}_{n}(\mathbf{x}_n) \cdot (\mathbf{x} - \mathbf{x}_n)} \xi_{n} G^{(\ell_n)}_{\lambda,n} }{\big[\one \!+\! T_N(-\lambda^2)\big]^{-1} \psi } \\
		= \braketl{\big[\one \!+\! T^{*}_N(-\lambda^2)\big]^{-1} \tx{\sum_{n,\ell_n}} q_n^{(\ell_n)} e^{-i \mathbf{S}_{n}(\mathbf{x}_n) \cdot (\mathbf{x} - \mathbf{x}_n)} \xi_{n} G^{(\ell_n)}_{\lambda,n} }{\psi } ,
}
which provides evidence for \eqref{eq: expG}. To say more, \eqref{eq: expG} can also be rephrased as
\bml{\label{eq: GTstar}
	\mathcal{G}(-\lambda^2) \mathbf{q}
	= \tx{\sum_{n,\ell_n}} q_n^{(\ell_n)} e^{-i \mathbf{S}_{n}(\mathbf{x}_n) \cdot (\mathbf{x} - \mathbf{x}_n)} \xi_{n} G^{(\ell_n)}_{\lambda,n} \\
		- T^{*}_N(-\lambda^2) \big[1 \!+\! T^{*}_N(-\lambda^2)\big]^{-1} \tx{\sum_{n,\ell_n}} q_n^{(\ell_n)} e^{-i \mathbf{S}_{n}(\mathbf{x}_n) \cdot (\mathbf{x} - \mathbf{x}_n)} \xi_{n} G^{(\ell_n)}_{\lambda,n}\,,
}
and the definition \eqref{eq: defLambda} yields, for any $ \mathbf{q} \in \mathbb{C}^{2N}$,
\begin{equation*}
 \tx\sum_{n,\ell_n'} \! \lf[\Lambda\big(\!-\lambda^2\big)\ri]_{mn}^{\ell_m \ell_n'} q_{n}^{(\ell_n')} = \tau_{m}^{(\ell_m)}\! \left[\mathcal{G}\big(\!-\lambda_0^2\big) \mathbf{q} - \mathcal{G}\big(\!-\lambda^2\big) \mathbf{q} \right].
\end{equation*}

On one side, by direct computations we deduce
\bmln{
	\tau_{m}^{(\ell_m)}\! \lf[ \tx\sum_{n,\ell_n'}  \lf\{ q_n^{(\ell_n')} e^{-i \mathbf{S}_{n}(\mathbf{x}_n) \cdot (\mathbf{x} - \mathbf{x}_n)} \xi_{n} G^{(\ell_n')}_{\lambda_0,n} -  q_n^{(\ell_n')} e^{-i \mathbf{S}_{n}(\mathbf{x}_n) \cdot (\mathbf{x} - \mathbf{x}_n)} \xi_{n} G^{(\ell_n')}_{\lambda,n} \ri\} \ri] \\
	= {\tx\sum_{\ell_m^{\prime}}} q_m^{(\ell_m^{\prime})}\! \lim_{r \to 0^{+}} \tfrac{i^{|\ell_m^{\prime} - \ell_m|}\,2^{|\ell_m + \alpha_m|-1} \,\Gamma\big(|\ell_m \!+\! \alpha_m|\big)}{r^{|\ell_m + \alpha_m|}}\; \times \\
		\times \big( |\ell_m + \alpha_m| + r\,\partial_{r}\big)\! \lf[ \lf(\lambda_0^{|\ell_m^{\prime} + \alpha_m|} K_{|\ell_m^{\prime} + \alpha_m|}(\lambda_0\,r) - \lambda^{|\ell_m^{\prime} + \alpha_m|} K_{|\ell_m^{\prime} + \alpha_m|}(\lambda\,r)\ri) J_{|\ell_m^{\prime} - \ell_m|}\big(\,|\mathbf{S}_{m}(\mathbf{x}_m)|\,r\big) \ri] \\
	= \tfrac{\pi}{2\,\sin(\pi \alpha_m)}\,\lf(\lambda^{2|\ell_m + \alpha_m|} - \lambda_0^{2|\ell_m + \alpha_m|} \ri) q_m^{(\ell_m)}\,.
}
On the other side, by \eqref{eq: TNdef} we get
\begin{equation}
	\tau_{m}^{(\ell_m)} T^{*}_N(-\lambda^2) \phi_{\lambda}
	= \braketr{G_{\lambda,m}^{(\ell_m)}}{P_{m}^{*}\, e^{i \mathbf{S}_{m}(\mathbf{x}_m) \cdot (\mathbf{x} - \mathbf{x}_m)} \phi_{\lambda}}
	= \braketl{e^{- i \mathbf{S}_{m}(\mathbf{x}_m) \cdot (\mathbf{x} - \mathbf{x}_m)} P_m\, G_{\lambda,m}^{(\ell_m)}}{ \phi_{\lambda} } . \label{eq: tauTstar}
\end{equation}
Summing up, the above arguments and few additional manipulations suffice to prove \eqref{eq: expLambda}.
\end{proof}

Considering the identity \eqref{eq: GTstar} derived in the previous proof, we introduce the bounded operator
\begin{equation}\label{eq: defFlam}
	\mathcal{F}_{\lambda} : \mathbb{C}^{2N} \to L^2(\mathbb{R}^2)\,, \qquad
	\mathcal{F}_{\lambda} \mathbf{q} := T^{*}_N(-\lambda^2) \big[1 \!+\! T^{*}_N(-\lambda^2)\big]^{-1} \tx\sum_{n,\ell_n} q_n^{(\ell_n)} e^{-i \mathbf{S}_{n}(\mathbf{x}_n) \cdot (\mathbf{x} - \mathbf{x}_n)} \xi_{n} G^{(\ell_n)}_{\lambda,n}\,.
\end{equation}
In particular, \eqref{eq: GTstar} reduces to
\begin{equation*}
	\mathcal{G}(-\lambda^2) \mathbf{q}
	= \tx\sum_{n,\ell_n} q_n^{(\ell_n)} e^{-i \mathbf{S}_{n}(\mathbf{x}_n) \cdot (\mathbf{x} - \mathbf{x}_n)} \xi_{n} G^{(\ell_n)}_{\lambda,n} - \mathcal{F}_{\lambda} \mathbf{q}\,.
\end{equation*}
	
	\begin{lemma}
		\label{lemma: FdomH}
		\mbox{}		\\
		For any $\lambda > 0$ fulfilling \eqref{eq: TN} and for all $ \mathbf{q}  \in \mathbb{C}^{2N}$, $	\mathcal{F}_{\lambda} \mathbf{q} \in \dom (H_{N}^{(\mathrm{F})} ) $ and
		\begin{equation}
				\big( H_{N}^{(\mathrm{F})} + \lambda^2\big) \mathcal{F}_{\lambda} \mathbf{q}
				= \tx{\sum_{n,\ell_n}} q_n^{(\ell_n)}\, e^{-i \mathbf{S}_{n}(\mathbf{x}_n) \cdot (\mathbf{x} - \mathbf{x}_n)} P_{n} G_{n}^{(\ell_n)}; \label{eq: HFFq}
			\end{equation}
			\begin{equation}\label{eq: tauFq}
				\lf[\bm{\tau} \mathcal{F}_{\lambda} \ri]_{m n}^{\ell_m \ell_n'}
				= \braketr{e^{- i \mathbf{S}_{m}(\mathbf{x}_m) \cdot (\mathbf{x} - \mathbf{x}_m)} P_m\, G_{\lambda,m}^{(\ell_m)}}{ \big[1 \!+\! T^{*}_N(-\lambda^2)\big]^{-1} e^{-i \mathbf{S}_{n}(\mathbf{x}_n) \cdot (\mathbf{x} - \mathbf{x}_n)} \xi_{n} G^{(\ell_n')}_{\lambda,n} } .
			\end{equation}
	\end{lemma}

\begin{proof}
In view of \eqref{eq: RNFried} and \eqref{eq: tauRGn}, we deduce
\bmln{
	\mathcal{F}_{\lambda} \mathbf{q} 
	= T^{*}_N(-\lambda^2) \big[1 \!+\! T^{*}_N(-\lambda^2)\big]^{-1} \tx{\sum_{n,\ell_n}} e^{-i \mathbf{S}_{n}(\mathbf{x}_n) \cdot (\mathbf{x} - \mathbf{x}_n)} \xi_{n} R^{(\mathrm{F})}_{n}(-\lambda^2) \xi_{n} e^{i \mathbf{S}_{n}(\mathbf{x}_n) \cdot (\mathbf{x} - \mathbf{x}_n)} \big(\tau^{(\ell_n)}_{n}\big)^{*} q_{n}^{(\ell_n)} \\
	= T^{*}_N(-\lambda^2)\, R^{(\mathrm{F})}_{N}(-\lambda^2) \tx{\sum_{n,\ell_n}} \big(\tau^{(\ell_n)}_{n}\big)^{*} q_{n}^{(\ell_n)}.
}
Notice that a direct computation gives $(H^{(\mathrm{F})}_{N} + \lambda^2)\, T_{N}^{*}(-\lambda^2) = T_{N}(-\lambda^2)\,(H^{(\mathrm{F})}_{N} + \lambda^2)$, which in turn implies
	\begin{equation*}
		T_{N}^{*}(-\lambda^2)\,R^{(\mathrm{F})}_{N}(-\lambda^2) = R^{(\mathrm{F})}_{N}(-\lambda^2)\, T_{N}(-\lambda^2)	\,.
	\end{equation*}
Taking this into account and using \eqref{eq: tauTstar}, for any $\phi \in L^2(\mathbb{R}^2)$ we obtain
\bmln{
	\braketl{\mathcal{F}_{\lambda} \mathbf{q}}{ \phi} 
	= \tx{\sum_{n,\ell_n}}\! \braketl{R^{(\mathrm{F})}_{N}(-\lambda^2) T_N(-\lambda^2) \big(\tau^{(\ell_n)}_{n}\big)^{*} q_{n}^{(\ell_n)}}{ \phi}
	= \tx{\sum_{n,\ell_n}} \big({q_{n}^{(\ell_n)}}\big)^* \, \tau^{(\ell_n)}_{n} T^{*}_N(-\lambda^2) R^{(\mathrm{F})}_{N}(-\lambda^2) \phi \\
	= \tx{\sum_{n,\ell_n}} \big({q_{n}^{(\ell_n)}}\big)^*\! \braketl{e^{- i \mathbf{S}_{n}(\mathbf{x}_n) \cdot (\mathbf{x} - \mathbf{x}_n)} P_n\, G_{\lambda,n}^{(\ell_n)}}{ R^{(\mathrm{F})}_{N}(-\lambda^2) \phi } ,
}
which entails
\begin{equation*}
\mathcal{F}_{\lambda} \mathbf{q} = R^{(\mathrm{F})}_{N}(-\lambda^2)\; \tx{ \sum_{n,\ell_n}} e^{- i \mathbf{S}_{n}(\mathbf{x}_n) \cdot (\mathbf{x} - \mathbf{x}_n)} P_n\, G_{\lambda,n}^{(\ell_n)} q_{n}^{(\ell_n)} .
\end{equation*}
This proves \eqref{eq: HFFq} and the fact that $ \mathcal{F}_{\lambda} \mathbf{q} $ belongs to the domain of the Friedrichs realization. The identity \eqref{eq: tauFq} follows straightforwardly from \eqref{eq: tauTstar} and \eqref{eq: defFlam}.
\end{proof}

In view of \cref{lemma: FdomH}, it is natural to wonder how the Hamiltonian operators $H_{N}^{(B)}$ and $H_{N}^{(\Theta)}$ are related.

\begin{proof}[Proof of \cref{prop:BTheta}]
We prove the thesis showing that $\dom \big(H_{N}^{(\Theta(B))}\big) = \dom \big(H_{N}^{(B)}\big)$, where $ \Theta(B) $ is the map \eqref{eq: TeB}, and $H_{N}^{(\Theta(B))} \psi = H_{N}^{(B)} \psi$, for any $\psi \in \dom \big(H_{N}^{(\Theta(B))}\big)$. 
On one hand, by  \cref{cor:domHbeta}, for any $\psi \in \dom \big(H_{N}^{(B)}\big)$ there exist $\phi_{\lambda} \in \dom\big(H_{N}^{(\mathrm{F})}\big)$ and $\mathbf{q} \in \mathbb{C}^{2N}$ such that
	\begin{gather*}
		\psi = \phi_{\lambda} + \mbox{$\sum_{n,\ell_n}$}\, q^{(\ell_n)}_{n}\, e^{-i \mathbf{S}_{n}(\mathbf{x}_n) \cdot (\mathbf{x} - \mathbf{x}_n)}\, \xi_{n}\, G_{\lambda,n}^{(\ell_n)}\,; \\
		\tau_{m}^{(\ell_m)}\phi_{\lambda} = \sum_{n,\ell_n} \lf[B^{(\ell_m \ell_n)}_{m\,n} + \tfrac{\pi\,\lambda^{2|\ell_n + \alpha_n|}}{ 2\sin(\pi\alpha_n)}\,\delta_{mn}\,\delta_{\ell_n \ell_n^{\prime}} \ri] q^{(\ell_n)}_{n}\,; \\
		\big( H_{N}^{(B)} + \lambda^2\big) \psi  
		= \big( H_{N}^{(\mathrm{F})} + \lambda^2\big) \phi_{\lambda}
			+ \tx\sum_{n,\ell_n}\! q^{(\ell_n)}_{n}\, e^{-i \mathbf{S}_{n}(\mathbf{x}_n) \cdot (\mathbf{x} - \mathbf{x}_n)}\, P_{n}\,G^{(\ell_n)}_{\lambda,n}\,.
	\end{gather*}	
On the other hand, for any $\psi \in \dom \big(H_{N}^{(\Theta)}\big)$, \cref{thm: extRF} yields that $ \psi = \varphi_{\lambda} + \mathcal{G}(-\lambda^2) \mathbf{q} $, for some suitable $\varphi_{\lambda} \in \dom\big(H_{N}^{(\mathrm{F})}\big)$ and $\mathbf{q} \in \mathbb{C}^{2N}$ with
	\begin{gather*}
		\tau_{m}^{(\ell_m)} \varphi_{\lambda} = \tx\sum_{n,\ell_n} \big[\Theta_{mn}^{\ell_m \ell_n} + \Lambda_{mn}^{\ell_m \ell_n}(-\lambda^2)\big] q_{n}^{(\ell_n)} \,;  \\
		\big(H_{N}^{(\Theta)} + \lambda^2 \big) \psi = \big(H_{N}^{(\mathrm{F})} + \lambda^2\big) \varphi_\lambda\,.
	\end{gather*}
In view of \cref{lemma: FdomH}, we may set
	\begin{equation}\label{eq: phiF}
		\phi_{\lambda} = \varphi_{\lambda} - \mathcal{F}_{\lambda} \mathbf{q}\,.
	\end{equation}
Then, using the boundary trace operator \eqref{eq: deftau}, together with the identities \eqref{eq: expLambda} and \eqref{eq: tauFq}, we deduce
	\bmln{
		0 = \tau_{m}^{(\ell_m)}\phi_{\lambda} - \tau_{m}^{(\ell_m)}\varphi_{\lambda} + \tau_{m}^{(\ell_m)} \mathcal{F}_{\lambda} \mathbf{q} \\
		= \tx\sum_{n,\ell_n'} \lf[B^{(\ell_m \ell_n')}_{m\,n} + \tfrac{\pi\,\lambda^{2|\ell_n + \alpha_n|}}{ 2\sin(\pi\alpha_n)}\,\delta_{mn}\,\delta_{\ell_m \ell_n^{\prime}} - \Theta_{mn}^{\ell_m \ell_n'} - \Lambda_{mn}^{\ell_m \ell_n'}(-\lambda^2) \ri. \\
			\lf. + \braketr{e^{- i \mathbf{S}_{m}(\mathbf{x}_m) \cdot (\mathbf{x} - \mathbf{x}_m)} P_m\, G_{\lambda,m}^{(\ell_m)}}{ \big[1 \!+\! T^{*}_N(-\lambda^2)\big]^{-1} e^{-i \mathbf{S}_{n}(\mathbf{x}_n) \cdot (\mathbf{x} - \mathbf{x}_n)} \xi_{n} G^{(\ell_n^{\prime})}_{\lambda,n} } \ri] q^{(\ell_n^{\prime})}_{n} \\
		= \tx\sum_{n,\ell_n^{\prime}} \lf[
				B^{(\ell_m \ell_n')}_{m\,n} - \Theta_{mn}^{\ell_m \ell_n'} + \tfrac{\pi\,\lambda_0^{2|\ell_n + \alpha_n|}}{ 2 \sin(\pi \alpha_n)}\,\delta_{mn}\delta_{\ell_m \ell_n^{\prime}} \ri. \\
				\lf. + \braketr{e^{- i \mathbf{S}_{m}(\mathbf{x}_m) \cdot (\mathbf{x} - \mathbf{x}_m)} P_m\, G_{\lambda_0,m}^{(\ell_m)}}{ \big[1 \!+\! T^{*}_N\big(-\lambda_0^2\big)\big]^{-1}  e^{-i \mathbf{S}_{n}(\mathbf{x}_n) \cdot (\mathbf{x} - \mathbf{x}_n)} \xi_{n} G^{(\ell_n^{\prime})}_{\lambda_0,n} }
				\ri] q^{(\ell_n^{\prime})}_{n} \,.
	}
In view of the arbitrariness of the charge $ \mathbf{q} \in \mathbb{C}^{2N}$, the above condition can be fulfilled only by fixing the Hermitian matrix $\Theta$ as in \eqref{eq: TeB}.
To say more, with the position \eqref{eq: phiF}, by means of \eqref{eq: HFFq} we obtain
	\bmln{
		\big( H_{N}^{(B)} + \lambda^2\big) \psi  
		= \big( H_{N}^{(\mathrm{F})} + \lambda^2\big) \varphi_{\lambda} 
			- \big( H_{N}^{(\mathrm{F})} + \lambda^2\big)\mathcal{F}_{\lambda} q
			+ \tx\sum_{n,\ell_n}\! q^{(\ell_n)}_{n}\, e^{-i \mathbf{S}_{n}(\mathbf{x}_n) \cdot (\mathbf{x} - \mathbf{x}_n)}\, P_{n}\,G^{(\ell_n)}_{\lambda,n} \\
		= \big( H_{N}^{(\mathrm{F})} + \lambda^2\big) \varphi_{\lambda} 
		= \big(H_{N}^{(\Theta(B))} + \lambda^2 \big) \psi\,,
	}
which concludes the proof.
\end{proof}

\subsection{Completion of the proofs}

We are finally in position to prove the main results about the singular realizations of the Schr\"odinger operator $ H_N $. Let us first complete the proof of \cref{cor:domHbeta}.
	
\begin{proof}[Proof of \cref{cor:domHbeta} - Part II]
It just remains to show that the family of operators $ H_N^{(B)} $, with $B \in M_{2N,\,\mathrm{Herm}}(\mathbb{C}) \cup \lf\{ \infty \ri\}$, exhausts all admissible self-adjoint realizations of the Schr\"odinger operator $H_{N}$. This is however a straightforward consequence of \cref{thm: extRF}, \cref{prop:BTheta} and \cite[Theorem 3.1 and Corollary 3.2]{Po08}. (Arguably, the same conclusion could also be derived adapting the arguments presented in \cite[Theorem 2.5]{BG85}.)
\end{proof}	

\begin{proof}[Proof of \cref{pro: spectrum}]
	\cref{thm: extRF} and \cref{prop:BTheta} ensure that the perturbed resolvent $R^{(B)}_{N}(z)$ is a finite range perturbation of the Friedrichs analogue $R^{(\mathrm{F})}_{N}(z)$. Then, $R^{(B)}_{N}(z) - R^{(\mathrm{F})}_{N}(z)$ is trace class and the thesis follows by Kuroda-Birman theorem \cite[Thm. XI.9]{RS81} (see also \cite{Ka57}).
	
	To say more, from \cite[Theorem 3.4]{Po04} we infer that the map $ \mathbf{q} \in \mathbb{C}^{2N} \mapsto \mathcal{G}(-\lambda^2)\,\mathbf{q} \in L^2(\mathbb{R}^{2})$ is a bijection from $ \ker \big[\Theta(B) + \Lambda(-\lambda^2) \big]$ onto $\ker \big(H_{N}^{(B)} + \lambda^2 \big)$. In other words, for any given negative eigenvalue $-\lambda^{2} \in \sigma_{\mathrm{disc}}\big(H_{N}^{(B)}\big)$, $\mathcal{G}(-\lambda^2)\,\mathbf{q} $ is an associated eigenvector for any $ \mathbf{q} \in \ker \big[\Theta(B) + \Lambda(-\lambda^2) \big]$.	
\end{proof}

\begin{proof}[Proof of \cref{cor:ScatteringHbeta}]
	The result follows from \cref{thm: OmHNH0} and the already mentioned Kuroda-Birman theorem, exploiting once more that $R^{(B)}_{N}(z)$ is a finite range perturbation of  $R^{(\mathrm{F})}_{N}(z)$.
\end{proof}

\end{document}